\documentclass[a4paper,fleqn]{cas-sc}

\usepackage{url} 
\usepackage[authoryear,longnamesfirst]{natbib}

\newtheorem{theorem}{Theorem}
\newtheorem{definition}{Definition}
\newtheorem{lemma}[theorem]{Lemma}
\newdefinition{rmk}{Remark}
\newproof{pf}{Proof}
\newproof{pot}{Proof of Theorem \ref{thm2}}
\def\tsc#1{\csdef{#1}{\textsc{\lowercase{#1}}\xspace}}
\tsc{WGM}
\tsc{QE}
\tsc{EP}
\tsc{PMS}
\tsc{BEC}
\tsc{DE}

\begin{document}
\let\WriteBookmarks\relax
\def\floatpagepagefraction{1}
\def\textpagefraction{.001}
\shorttitle{Modeling Biocontrol of soybean pod borer (\textit{Leguminivora glycinivorella})}
\shortauthors{Wenxuan Li et~al.}

\title [mode = title]{
A data-driven stage-structured host-parasitoid model for optimizing \textit{Trichogramma} interventions against soybean pod borer (\textit{Leguminivora glycinivorella}) outbreaks}

\author[1]{Wenxuan Li}[style=chinese, orcid=0000-0002-3601-117X]
\ead{92530@jssvc.edu.cn}

\credit{Writing – original draft, Visualization, Validation,
Software, Methodology, Investigation, Formal analysis, Data curation}

\affiliation[1]{organization={Department of Mathematics and Physics},
                addressline={Suzhou Polytechnic University}, 
                city={Suzhou},
                postcode={215104}, 
                country={China}}

\author[2]{Xu Chen}[style=chinese, orcid=0000-0002-0866-6857]
\ead{chenxu21@mails.jlu.edu.cn}

\credit{Writing – original draft, Visualization, Validation,
Software, Methodology, Investigation, Formal analysis, Data curation}

\affiliation[2]{organization={School of Mathematics},
                addressline={Changchun University of Technology}, 
                city={Changchun},
                postcode={130012}, 
                country={China}}

\author[3]{Yu Gao}[style=chinese, orcid=0000-0002-4354-2665]
\cormark[1]
\ead{gaothrips@jlau.edu.cn}

\credit{Writing – review \& editing, Writing –
original draft, Supervision, Resources, Project administration, Funding
acquisition, Conceptualization}

\affiliation[3]{organization={College of Plant Protection},
                addressline={Jilin Agricultural University}, 
                city={Changchun},
                postcode={130118}, 
                country={China}}

\author[4]{Suli Liu}[style=chinese, orcid=0000-0002-3369-8578]
\cormark[2]
\ead{liusuli@jlu.edu.cn}

\credit{Writing – review \& editing, Writing –
original draft, Supervision, Resources, Project administration, Funding
acquisition, Conceptualization}

\affiliation[4]{organization={School of Mathematics},
                addressline={Jilin University}, 
                city={Changchun},
                postcode={130012}, 
                country={China}}

\cortext[cor1]{Corresponding author. No. 2888, Xincheng Street, Nanguan District, Changchun, 130118, China}
\cortext[cor1]{Corresponding author. No. 2699, Qianjin Street, Chaoyang District, Changchun, 130012, China}

\begin{abstract}
The soybean pod borer (\textit{Leguminivora glycinivorella}) poses a severe threat to global soybean production. 
In this study, we developed a stage-structured host-parasitoid dynamic model that explicitly couples the holometabolous life cycle of the pest with the obligate egg-parasitism mechanism of \textit{Trichogramma} wasps. 
Utilizing field monitoring data from Changchun, Jilin Province, key biological parameters were rigorously estimated via the Markov Chain Monte Carlo (MCMC) method. 
This calibration facilitated the establishment of a precise Economic Injury Level ($Q_{EIL}$) of 0.0389 individuals/$m^2$, based solely on the destructive larval stage. 
Through theoretical and numerical analyses of different intervention scenarios, we identified an optimal continuous release rate ($C^* = 2.645$) that efficiently suppresses the outbreak without causing wasteful parasitoid accumulation. 
Furthermore, simulations demonstrate that a 5-day impulsive release interval provides the optimal balance between strict pest suppression and field operational costs. 
This study bridges the gap between theoretical population dynamics and applied agricultural management, providing a directly applicable mathematical decision-making tool for the precise biological control of crop pests.  
\end{abstract}

\begin{graphicalabstract}
\includegraphics[width=\linewidth]{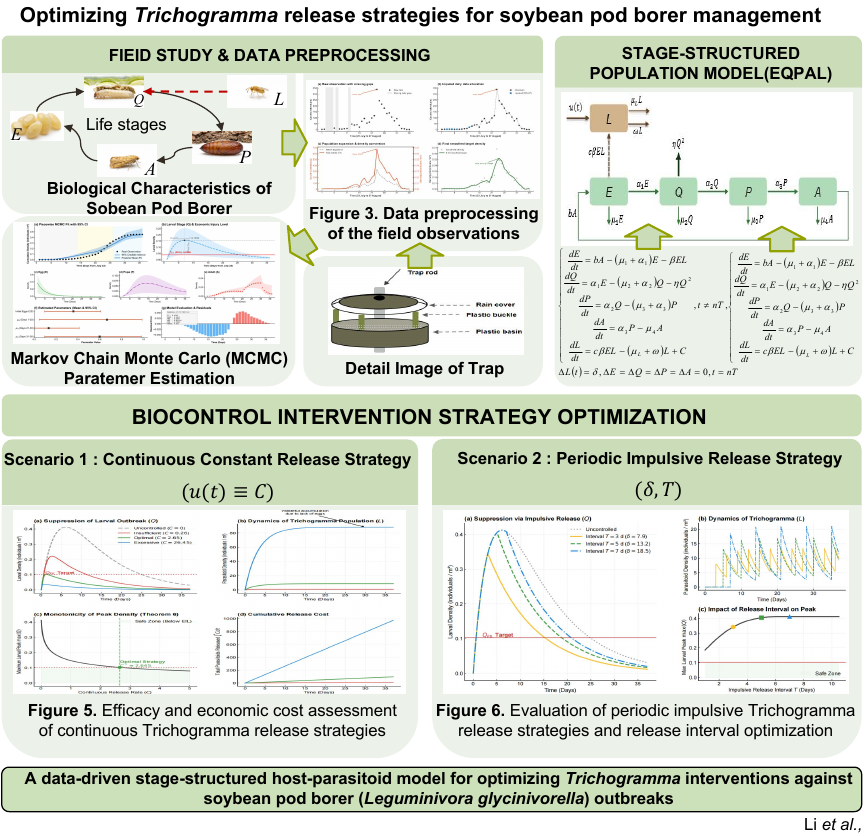}
\end{graphicalabstract}

\begin{highlights}
\item Established the first stage-structured compartmental model specifically for the soybean pod borer and \textit{Trichogramma}.
\item Applied the MCMC method to fit field monitoring data, successfully inferring key biological parameters.
\item Optimized practically viable \textit{Trichogramma} release strategies that balance pest suppression efficacy with economic costs.
\end{highlights}


\maketitle

\section{Introduction}\label{Sec1}
Soybean is a core global crop with triple roles as a food source, oilseed, and feedstock \citep{hartman2011crops}.
However, the soybean pod borer poses a significant threat to global soybean production \citep{swaminathan2013food}. 
Its geographical distribution is primarily concentrated in East Asia, including China, Japan, North Korea, South Korea, and the Russian Far East (Figure \ref{Figure1}(a)). 
Among these regions, China is the most severely affected, with the pest predominantly distributed across the Northeast, North, Northwest, and Huang-Huai-Hai soybean-producing regions (Figure \ref{Figure1}(b)). 
Notably, the Northeast spring soybean region—encompassing Heilongjiang, Jilin, and Liaoning provinces—alongside Hebei and Shandong, are identified as high-risk, severe outbreak areas. 
Recent habitat suitability predictions indicate that the pest's potential distribution is expanding, further escalating its threat to these critical soybean-producing areas \citep{gao2018current,yang2024predicting}.
This pest inflicts damage by boring into soybean pods and feeding on developing seeds as larvae. It not only causes direct yield losses of 10\%–30\% but also significantly reduces the commercial quality and oil content of harvested soybeans \citep{gaur2018pests}. 
Although conventional chemical control can rapidly reduce pest populations in the short term, its long-term overuse has led to a series of problems, including increased pest resistance, decline of natural enemy populations, and imbalance in farmland ecosystems \citep{ali2023towards}. 
Therefore, developing green biological control technologies centered on natural enemy utilization has become an urgent requirement for modern soybean pest management systems \citep{maciel2022role}.

\textit{Trichogramma} spp. are obligate egg parasitoids of the soybean pod borer. 
They have distinct advantages including high host specificity, rapid reproduction, and safety to humans, livestock, and the environment \citep{consoli2010egg}. 
They are currently one of the most widely used biological control agents worldwide \citep{smith1996biological}. 
However, \textit{Trichogramma} release strategies in major soybean-producing regions of China still rely heavily on empirical judgment. 
There is a lack of quantitative scientific basis for key parameters such as release timing, release frequency, and single release dosage. 
This leads to highly variable field control efficacy and relatively high control costs, which severely restricts the large-scale application of this technology. 
The root cause of this dilemma lies in the limitations of existing research. 
No mathematical model has been established to accurately describe the full-generation developmental patterns of the soybean pod borer and its parasitic interaction mechanisms with \textit{Trichogramma}. 
As a result, the regulatory effects of different release strategies on pest populations cannot be revealed from a dynamical perspective.

Mathematical population dynamics models are core quantitative tools \citep{li2026dynamic,tan2026stochastic,liu2023discrete,wu2026rapid,xie2025tentative,peng2024predicting}. 
They reveal the evolutionary patterns of agricultural and forestry pest populations, quantify the interaction mechanisms between natural enemies and pests, and optimize biological control strategies. 
These models have been widely used in the prediction, forecasting, and control decision-making of various agricultural pests \citep{ibrahim2022expert,tonnang2017advances}.

\begin{figure}[t]
\centering
\includegraphics[width=\textwidth]{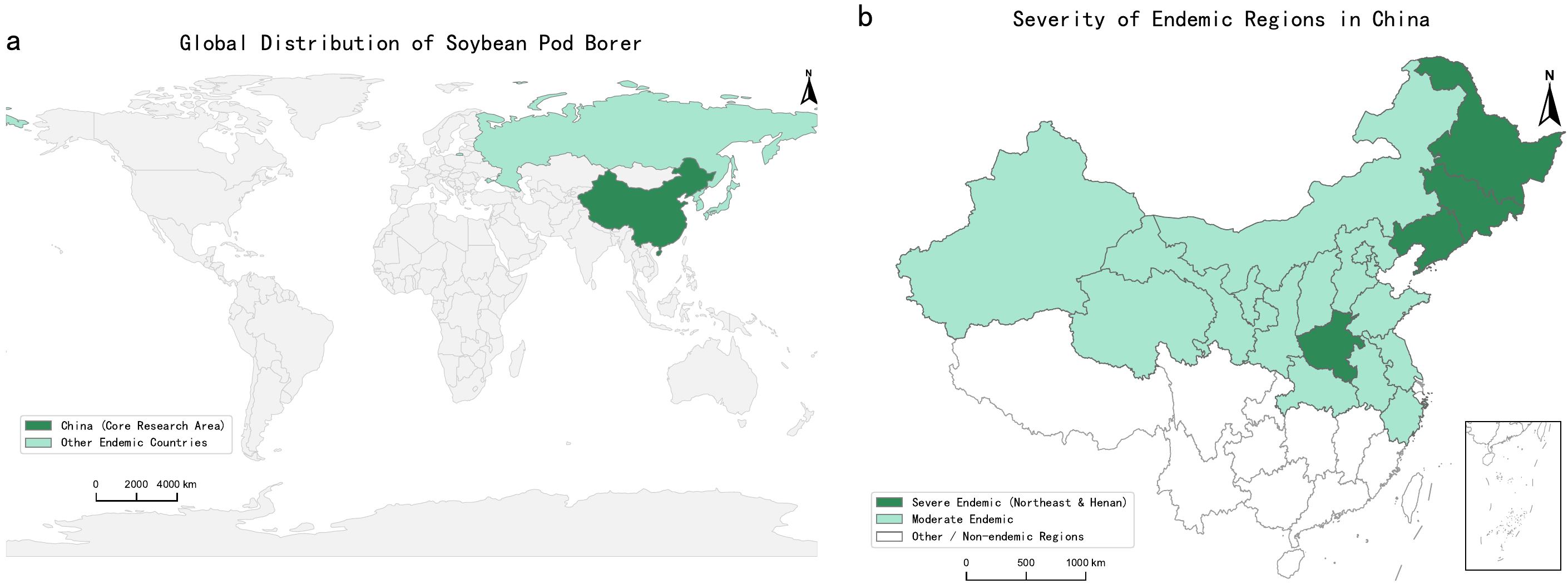}
\caption{\textbf{Geographical distribution and endemic severity of the soybean pod borer.}
 (\textbf{a}) Global distribution of the pest, highlighting China as the core research area (dark green) alongside other endemic countries (light green). 
 (\textbf{b}) Severity of endemic regions in China. Severe endemic regions, including Northeast China and Henan, are highlighted in dark green, while moderate endemic regions are shown in light green.
 The administrative boundaries are derived from the Natural Earth dataset (\url{naturalearthdata.com}) and the DataV.GeoAtlas platform (\url{datav.aliyun.com/portal/school/atlas/area_selector}).
\label{Figure1}}
\end{figure}   

Therefore, this study aims to answer two critical questions: 
\begin{itemize}
    \item What are the intrinsic baseline dynamics of soybean pod borer under natural intra-specific competition, and how do they dictate the intervention thresholds?
    \item How can we optimize the continuous and impulsive release rates of Trichogramma to balance pest suppression with economic costs?
\end{itemize}
To address these scientific gaps, we developed the EQPAL stage-structured impulsive differential equation model. 
This model describes the interaction between the soybean pod borer and \textit{Trichogramma}, covering four developmental stages: eggs, larvae, pupae, and adults.
It is the first model to accurately couple two key mechanisms: the quadratic density-dependent mortality of soybean pod borer larvae, and the obligate egg parasitism of \textit{Trichogramma}. 
We simulated three distinct control scenarios: no artificial release, continuous constant release, and periodic impulsive release.
Through rigorous mathematical derivations, we proved the existence and global asymptotic stability of system equilibria under each scenario. 
We also derived the critical threshold conditions for pest extinction and persistence.
Meanwhile, we used systematic field monitoring data of the soybean pod borer collected in Changchun in 2025. 
We applied the Markov Chain Monte Carlo (MCMC) method to invert and estimate the key biological parameters of the model.
We quantified the EIL of the soybean pod borer. 
Finally, we optimized the optimal \textit{Trichogramma} release strategy that balances control efficacy and economic cost.

\begin{figure}[t]
\centering
\includegraphics[width=\textwidth]{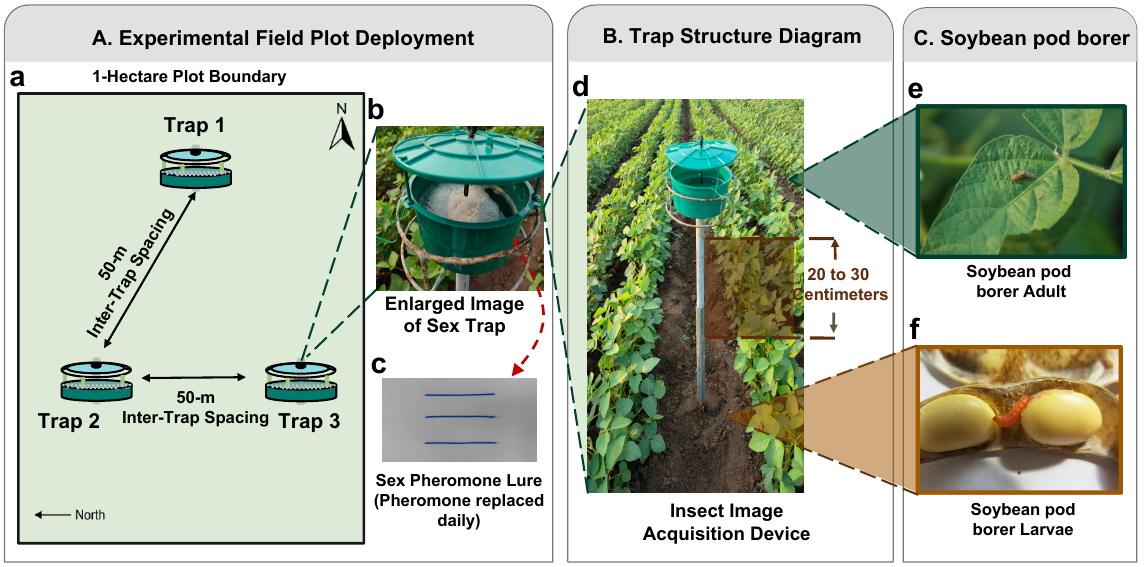}
\caption{\textbf{Field experimental setup and biological characteristics of the soybean pod borer.} 
    (\textbf{a}) Schematic diagram of the 1-hectare soybean field layout, indicating the 1-hectare boundary and the installation locations of three traps. The distance between adjacent traps is 50 \textit{m}. 
    (\textbf{b}) Enlarged photograph of a barrel-shaped sex pheromone trap. 
    (\textbf{c}) Photograph of the sex pheromone lure; the lure pheromone was replaced daily. B. Trap structure diagram 
    (\textbf{d}) Photograph of the sex pheromone trap device, showing the trap height set at 20 to 30 \textit{cm} above the soybean canopy. 
    (\textbf{e}) Close-up of an adult soybean pod borer resting on a soybean leaf. 
    (\textbf{f}) Close-up of a soybean pod borer larva inside an opened soybean pod.\label{Figure2}}
\end{figure}   

This study makes two key contributions. 
Theoretically, it improves the stage-structured impulsive dynamics theory for agricultural and forestry pest-parasitoid interaction systems. 
It also clarifies the intrinsic mechanisms by which different release strategies regulate pest populations.
Practically, it provides a directly applicable quantitative decision-making tool for the green and precise control of the soybean pod borer.
The rest of the paper is organized as follows. 
First, we describe the field data collection and preprocessing methods. Next, we present the construction and dynamical analysis of the EQPAL model. 
Then, we show the results of model parameter fitting and validation. We further compare the control efficacy of different release strategies. 
Finally, we summarize the main conclusions and outline future research directions.

\section{Materials and methods}\label{Sec2}
\subsection{Data collection and preprocessing}\label{Sec2.1}
We conducted field experiments in a 1-hectare continuous soybean plot at the Soybean Regional Technical Innovation Center of Jilin Agricultural University (Changchun, China; $125.42^\circ\text{E}, 43.82^\circ\text{N}$).
To capture the authentic population dynamics of soybean pod borer, the plot was maintained entirely free of chemical pesticides throughout the crop's growth period.

We performed daily monitoring of adult male moths from July 23 to August 27, 2025, completely covering the critical emergence and oviposition windows. We deployed three basin-type sex pheromone traps (Figure\ref{Figure2}(b)) with a minimum 50-m spatial interval to prevent pheromone plume interference. 
We filled the plastic basins with a diluted detergent solution to reduce surface tension for effective capture. We suspended sex pheromone lures centrally and replaced them once a day to maintain stable attractant efficacy.

We inspected the traps and counted captured moths daily at dawn. Because the pheromone lures exclusively attract male moths, the raw count data, denoted as $A_{trap}(t)$, represent only half of the population. Assuming a 1:1 natural sex ratio for soybean pod borer, we calculated the total daily adult emergence $A_{total}(t)$ as follows: 
\begin{figure}[h]
\centering
\includegraphics[width=0.65\textwidth]{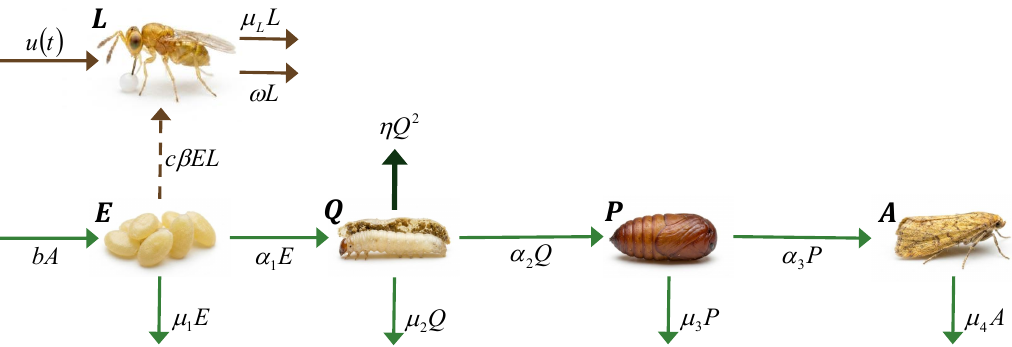}
\caption{Flow diagram for system \eqref{equation3}.\label{Figure3}}
\end{figure}   
\unskip
\begin{equation}\label{equation1}
A_{total}(t) = 2 \times A_{trap}(t)
\end{equation}
\subsection{Model and formulation}\label{Sec2.2}
To describe the interaction between the soybean pod borer and its obligate egg parasitoid, \textit{Trichogramma}, we developed a stage structured host-parasitoid dynamic model. 
The model incorporates the soybean pod borer life cycle (egg, larva, pupa, adult) of the borer and the specific egg-parasitism mechanism of the \textit{Trichogramma}. 
Our objective is to facilitate the quantitative evaluation of biological control efficacy using artificial \textit{Trichogramma} releases to manage soybean pod borer outbreaks.

To ensure the mathematical model strictly conforms to the realistic ecological principles of agricultural fields, the proposed model is formulated based on the following core assumptions:
\begin{enumerate}
    \item \textbf{Stage-structured population and complete metamorphosis}: 
    The soybean pod borer is a typical holometabolous insect \citep{chen2025role}. 
    Its life cycle is strictly divided into four sequential developmental stages: eggs $E(t)$, larvae $Q(t)$, pupae $P(t)$, and adults $A(t)$. 
    The population density of the natural enemy \textit{Trichogramma} is denoted by $L(t)$. 
    It is assumed that all individuals are uniformly distributed within their respective developmental stages.
    \item \textbf{Adult reproduction and non-overlapping resources}: 
    Unlike hemimetabolous insects, the primary ecological functions of adult pod borers are mating and oviposition. 
    They do not feed on soybean plant tissues and therefore do not participate in core resource competition within the agricultural field. 
    It is assumed that adults oviposit at a constant rate $b$, and follow a natural senescence process with a constant mortality rate $\mu_4$.
    \item \textbf{Intraspecific competition and density-dependent mortality in larvae}: 
    After hatching, larvae $Q(t)$ must bore into soybean pods to feed on the developing seeds. 
    Due to the absolute upper limit of pod numbers and spatial resources per soybean plant, an increase in larval density will lead to intense intraspecific competition. 
    Therefore, this model abandons the traditional system-wide Logistic growth assumption and precisely incorporates the density-dependence mechanism into the larval compartment. 
    A quadratic mortality term $-\eta Q^2$ is introduced to represent the additional mortality caused by resource depletion, where $\eta$ is the intraspecific competition coefficient.
    \item \textbf{Exclusive egg-parasitism mechanism}: 
    \textit{Trichogramma} $L(t)$ is an obligate egg parasitoid. 
    It is assumed that parasitoids randomly search for and parasitize soybean pod borer eggs $E(t)$ in the field, following the mass-action principle with a parasitism rate of $\beta EL$. 
    Parasitized eggs immediately cease development and hatch into the next generation of \textit{Trichogramma} under a conversion rate $c$. 
    Parasitoids have no direct impact on larvae, pupae, or adults.
     \item \textbf{Artificial biocontrol intervention}: 
     To simulate green biocontrol measures in practical agricultural production, it is assumed that \textit{Trichogramma} are artificially released into the field at a rate $u(t)$. 
     Simultaneously, parasitoids are affected by the natural environment, exhibiting natural mortality and emigration behaviors, with a combined elimination rate denoted as $\mu_L + \omega$. 
\end{enumerate}

The total number of pests at time $t$, denoted by $N(t)$, is given by:
\begin{equation}\label{equation2}
N(t)=E(t)+Q(t)+P(t)+A(t)+L(t).
\end{equation}

All state variables, physical meanings of biological parameters, values, and corresponding references involved in this model are detailed in Table \ref{Table1}.

\begin{table}[width=\linewidth,cols=5,pos=h]
\caption{State variables and parameter definitions for the EQPAL model.}\label{Table1}
\renewcommand{\arraystretch}{1.1}
\begin{tabular*}{\tblwidth}{@{} LLLLL@{} }
\toprule
\textbf{Symbol} & \textbf{Biological Definition} & \textbf{Value} & \textbf{Unit} & \textbf{Reference} \\
\midrule
\textbf{Variables} &&&&\\
$E(t)$ & Density of soybean pod borer eggs & 0.95 & Individuals/$m^2$ & Estimated \\
$Q(t)$ & Density of soybean pod borer larvae & 0 & Individuals/$m^2$ & Assumed \\
$P(t)$ & Density of soybean pod borer pupae & 0 & Individuals/$m^2$ & Assumed \\
$A(t)$ & Density of soybean pod borer adults & 0 & Individuals/$m^2$ & Assumed \\
$L(t)$ & Density of \textit{Trichogramma} & 0 & Individuals/$m^2$ & Assumed \\
\hline
\textbf{Parameters} &&&&\\
$b$ & Effective oviposition rate of adults & 0.01 & Day$^{-1}$ & Assumed \\
$\alpha_1$ & Hatching rate from eggs to larvae & 0.20 & Day$^{-1}$ & \citeyear{cea2015website} \\
$\alpha_2$ & Pupation rate from larvae to pupae & 0.08 & Day$^{-1}$ & \citeyear{shenyang2022website} \\
$\alpha_3$ & Emergence rate from pupae to adults & 0.12 & Day$^{-1}$ & \citeyear{foodmate2006website} \\
$\mu_1$ &  Mortality rate of eggs & 0.05 & Day$^{-1}$ & \citeyear{cea2015website} \\
$\mu_2$ &  Mortality rate of larvae & 0.02 & Day$^{-1}$ & \citep{hsu1965study} \\
$\mu_3$ &  Mortality rate of pupae & 0.02 & Day$^{-1}$ & \citeyear{shenyang2022website} \\
$\mu_4$ & Mortality rate of adults & 0.25;0.05;0.40 & Day$^{-1}$ & Estimated \\
$\eta$ & Intraspecific competition coefficient & 0.001 & (Ind/$m^2$)$^{-1}$ Day$^{-1}$ & Assumed \\
$\beta$ & Parasitism search rate & 1.2 & (Ind/$m^2$)$^{-1}$ Day$^{-1}$ & Assumed \\
$c$ & Conversion rate of parasitized eggs & 1.0 & Dimensionless & Assumed \\
$\mu_L$ &  Mortality rate of \textit{Trichogramma} & 0.25 & Day$^{-1}$ & \citep{gurr2000effect} \\
$\omega$ & Emigration rate of \textit{Trichogramma} & 0.05 & Day$^{-1}$ & Assumed \\
$u(t)$ & Artificial release rate function & C & Individuals/$m^2 $Day$^{-1}$ & Assumed \\
\bottomrule
\end{tabular*}
\end{table}

Based on the aforementioned biological assumptions, the population dynamics of the soybean pod borer and \textit{Trichogramma} interacting system are governed by the following nonlinear ordinary differential equations:
\begin{equation}
\label{equation3}
\begin{cases}
\frac{\mathrm{d}E}{\mathrm{d}t} = bA - (\alpha_1 + \mu_1)E - \beta EL \\[2pt]
\frac{\mathrm{d}Q}{\mathrm{d}t} = \alpha_1 E - (\alpha_2 + \mu_2)Q - \eta Q^2 \\[2pt]
\frac{\mathrm{d}P}{\mathrm{d}t} = \alpha_2 Q - (\alpha_3 + \mu_3)P \\[2pt]
\frac{\mathrm{d}A}{\mathrm{d}t} = \alpha_3 P - \mu_4 A \\[2pt]
\frac{\mathrm{d}L}{\mathrm{d}t} = c\beta EL - (\mu_L + \omega)L + u(t)
\end{cases}
\end{equation}
The initial conditions of this system satisfy \(E(0) \ge 0\), \(Q(0) \ge 0\), \(P(0) \ge 0\), \(A(0) \ge 0\) and \(L(0) \ge 0\). Let \(N(t) = E(t) + Q(t) + P(t) + A(t) + L(t)\) represent the total population size of the system at time $t$. 
Figure \ref{Figure3} depicts the compartmental transition flows of the proposed model.

\section{Results}\label{Sec3}
\subsection{Model analysis}\label{Sec3.1}
To further investigate the biocontrol potential of \textit{Trichogramma} against the soybean pod borer, this section systematically analyzes the dynamical behaviors of the proposed EQPAL model. 
Specifically, we explore three distinct scenarios: no artificial release, a continuous constant release strategy, and an impulsive periodic release strategy.
\subsubsection{Dynamical behaviors without artificial release ($u(t) = 0$)}\label{Sec3.1.1}
In the absence of continuous artificial releases of natural enemies in the agricultural environment, the interactive system between the soybean pod borer and \textit{Trichogramma} degenerates into an autonomous dynamical system. 
In this case, system \eqref{equation3} simplifies to the following system:
\begin{equation}
\label{equation4}
\begin{cases}
\frac{\mathrm{d}E}{\mathrm{d}t} = bA - (\alpha_1 + \mu_1)E - \beta EL \\[2pt]
\frac{\mathrm{d}Q}{\mathrm{d}t} = \alpha_1 E - (\alpha_2 + \mu_2)Q - \eta Q^2 \\[2pt]
\frac{\mathrm{d}P}{\mathrm{d}t} = \alpha_2 Q - (\alpha_3 + \mu_3)P \\[2pt]
\frac{\mathrm{d}A}{\mathrm{d}t} = \alpha_3 P - \mu_4 A \\[2pt]
\frac{\mathrm{d}L}{\mathrm{d}t} = c\beta EL - (\mu_L + \omega)L
\end{cases}
\end{equation}

This system describes the natural reproduction of the pod borer under the limitation of soybean pod resources, as well as the natural parasitism process carried out by the wild \textit{Trichogramma} population. 
Table \ref{Table2} summarizes the core definitions, lemmas, and theorems to be proven in this subsection, along with their corresponding biological interpretations. 
All mathematical proof processes are detailed in Appendix \ref{appA}.

\begin{table}[width=\linewidth,cols=4,pos=h]
\caption{Core lemmas and theorems of the model dynamic analysis without artificial release.}\label{Table2}
\begin{tabular*}{\tblwidth}{@{} LLLL@{} }
\toprule
 Category & Identifier & Core Conclusion & Appendix Reference\\
\midrule
Theorem & Theorem \ref{th1} & Non-negativity and boundedness of the model’s state variables & Appendix \ref{app1} \\
Theorem & Theorem \ref{th2} & Existence and calculation of the Parasitoid-Free Equilibrium $E_0$ & Appendix \ref{app2} \\
Lemma & Lemma \ref{le3} & Derivation of the basic reproduction number $\mathcal{R}_0$ & Appendix \ref{app3} \\
Definition & Definition \ref{de1} & Local asymptotic stability of the Parasitoid-Free Equilibrium $E_0$ & Appendix \ref{app4} \\
Theorem & Theorem \ref{th4} & Global asymptotic stability of the Parasitoid-Free Equilibrium $E_0$ & Appendix \ref{app5} \\
Lemma & Lemma \ref{le5} & Existence and uniqueness of the Coexistence Equilibrium $E^*$ & Appendix \ref{app6} \\
Theorem & Theorem \ref{th6} & Global asymptotic stability of the Coexistence Equilibrium $E^*$ & Appendix \ref{app7} \\
\bottomrule
\end{tabular*}
\end{table}
\begin{table}[width=\linewidth,cols=4,pos=h]
\caption{Core lemmas and theorems for the model with continuous constant release.}\label{Table3}
\begin{tabular*}{\tblwidth}{@{} LLLL@{} }
\toprule
 Category & Identifier & Core Conclusion & Appendix Reference\\
\midrule
Lemma & Lemma \ref{le7} & Non-existence of the parasitoid-free equilibrium & Appendix \ref{app8} \\
Lemma & Lemma \ref{le8} & Existence and uniqueness of the coexistence equilibrium $E^*_C$ & Appendix \ref{app9} \\
Theorem & Theorem \ref{th9} & Global asymptotic stability of the coexistence equilibrium $E^*_C$ & Appendix \ref{app10} \\
Theorem & Theorem \ref{th10} & Monotonicity of pest suppression with respect to release rate $C$ & Appendix \ref{app11} \\
\bottomrule
\end{tabular*}
\end{table}
\begin{table}[width=\linewidth,cols=4,pos=h]
\caption{Core lemmas and theorems for the model with periodic impulsive release.}\label{Table4}
\begin{tabular*}{\tblwidth}{@{} LLLL@{} }
\toprule
Category & Identifier & Core Conclusion & Appendix Reference\\
\midrule
Lemma & Lemma \ref{le11} & Existence and stability of the pest-free periodic solution & Appendix \ref{app12} \\
Theorem & Theorem \ref{th12} & Global stability of the pest-free periodic solution & Appendix \ref{app13} \\
Theorem & Theorem \ref{th13} & Threshold condition for the permanence of the pest population & Appendix \ref{app14} \\ 
\bottomrule
\end{tabular*}
\end{table}

\subsubsection{Dynamical behaviors under continuous constant release strategy ($u(t) \equiv C$)}\label{Sec3.1.2}

In practical biological control, continuous release of natural enemies is a prevalent strategy to maintain constant predatory pressure and suppress pest populations below the economic injury level \citep{van2000success}. 
In this subsection, we assume the artificial release rate is a constant $C>0$. 
Consequently, the autonomous system \eqref{equation3} is extended to the following non-homogeneous system:
\begin{equation}\label{equation5}
\begin{cases} 
\frac{\mathrm{d}E}{\mathrm{d}t}=bA-(\alpha_{1}+\mu_{1})E-\beta EL \\[2pt]
\frac{\mathrm{d}Q}{\mathrm{d}t}=\alpha_{1}E-(\alpha_{2}+\mu_{2})Q-\eta Q^{2} \\[2pt]
\frac{\mathrm{d}P}{\mathrm{d}t}=\alpha_{2}Q-(\alpha_{3}+\mu_{3})P \\[2pt]
\frac{\mathrm{d}A}{\mathrm{d}t}=\alpha_{3}P-\mu_{4}A \\[2pt]
\frac{\mathrm{d}L}{\mathrm{d}t}=c\beta EL-(\mu_{L}+\omega)L + C 
\end{cases}
\end{equation}


Unlike the natural system, the presence of the constant term $C$ ensures the persistence of the parasitoid population even when the prey egg density is extremely low. 
Table \ref{Table3} summarizes the core lemmas and theorems for this scenario, and the detailed mathematical proof processes are provided in Appendix \ref{appB}.

\subsubsection{Dynamical behaviors under periodic impulsive release strategy}\label{Sec3.1.3}

In practical agricultural management, the continuous release of natural enemies is often cost-prohibitive and technically challenging to implement in large-scale soybean fields \citep{martinez2025baculoviruses}. 
A more feasible and widely adopted approach is the periodic impulsive release strategy, where a specific dosage of \textit{Trichogramma} is released at discrete, fixed time intervals. 
To characterize this non-continuous biocontrol intervention, we extend the autonomous system \eqref{equation3} into the following impulsive stage-structured differential system:

\begin{equation}\label{equation6}
\begin{cases} 
\begin{aligned}
&\left. 
\begin{aligned}
\frac{\mathrm{d}E}{\mathrm{d}t} &= bA - (\alpha_1 + \mu_1)E - \beta EL \\[1.5pt]
\frac{\mathrm{d}Q}{\mathrm{d}t} &= \alpha_1 E - (\alpha_2 + \mu_2)Q - \eta Q^2 \\[1.5pt]
\frac{\mathrm{d}P}{\mathrm{d}t} &= \alpha_2 Q - (\alpha_3 + \mu_3)P \\[1.5pt]
\frac{\mathrm{d}A}{\mathrm{d}t} &= \alpha_3 P - \mu_4 A \\[1.5pt]
\frac{\mathrm{d}L}{\mathrm{d}t} &= c\beta EL - (\mu_L + \omega)L
\end{aligned} \right\} t \neq nT, \\[3pt]
&\left. \Delta L(t) = \delta, \quad \Delta E = \Delta Q = \Delta P = \Delta A = 0 \right\} t = nT,
\end{aligned}
\end{cases}
\end{equation}

where $T > 0$ denotes the constant release period, and $\delta > 0$ represents the impulsive release dosage of \textit{Trichogramma} at each release time $nT$ ($n \in \mathbb{N}$). 
The state variable $L(t^+)$ signifies the instantaneous density of the parasitoid population immediately after the $n$-th impulsive release event.

The primary objective of this subsection is to establish the threshold conditions under which the soybean pod borer population can be eradicated or effectively suppressed below the EIL using this intermittent strategy. 
Unlike the previous autonomous or non-homogeneous systems, system \eqref{equation6} exhibits complex impulsive periodic oscillations. 
Table \ref{Table4} summarizes the core definitions, lemmas, and stability theorems for the impulsive periodic solution, with comprehensive mathematical proofs provided in Appendix \ref{appC}.

\subsection{Model fitting}\label{Sec3.2}
\subsubsection{Field data preprocessing}
\begin{figure}[t]
\centering
\includegraphics[width=0.9\linewidth]{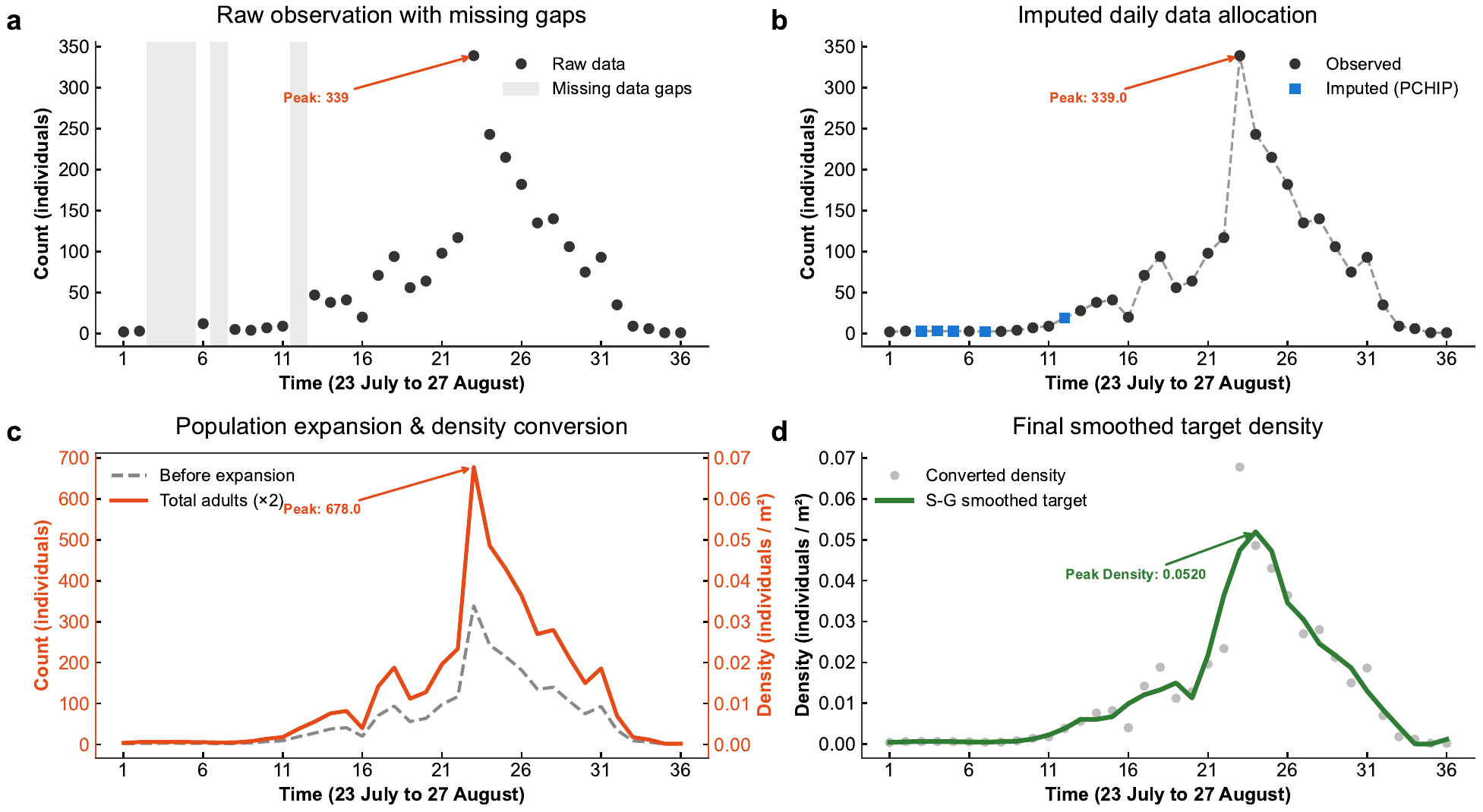}
\caption{\textbf{Data preprocessing of the field observations.} 
    (\textbf{a}) \textbf{Raw observation with missing gaps}. Raw daily catch data of male soybean pod borer adults, with shaded areas indicating missing observation gaps.
    (\textbf{b}) \textbf{Imputed daily data allocation}. Imputed daily data allocation using a Piecewise Cubic Hermite Interpolating Polynomial (PCHIP). 
    (\textbf{c}) \textbf{Population expansion \& density conversion}. Pest expansion and density conversion. 
    (\textbf{d}) \textbf{Final smoothed target density}. Final smoothed target density processed via a Savitzky-Golay filter. \label{Figure4}}
\end{figure}   
\unskip

Prior to the Markov Chain Monte Carlo (MCMC) parameter estimation, we processed the raw field monitoring data through a rigorous mathematical transformation pipeline.
As illustrated in Figure \ref{Figure4}(a), while the monitoring period spanned from July 23 to August 27, 2025, unavoidable field constraints resulted in several observational gaps. 
Consequently, certain records represented the cumulative catch over multiple days rather than daily increments. 
To satisfy the requirement for temporal continuity and conform to the biological principles of natural population emergence, we employed a Piecewise Cubic Hermite Interpolating Polynomial (PCHIP) \citep{fritsch1980monotone} to reconstruct the cumulative capture trajectory (Figure \ref{Figure4}(b)). 
Let $C(t_k)$ denote the cumulative catch recorded on observation day $t_k$; the PCHIP algorithm constructs a cubic polynomial $H(t)$ that preserves local monotonicity. 
The continuous daily increment $D(t)$ can then be derived as $D(t) = H(t) - H(t-1)$. 
This method effectively imputes missing values while maintaining the authentic trend of monotonic population growth and avoiding artificial numerical oscillations common in traditional interpolation techniques.

Given that sex pheromone traps specifically attract only male adults \citep{rizvi2021latest}, we adjusted the single-sex capture data based on the biological evidence of a 1:1 natural sex ratio insoybean pod borer populations. 
We converted the absolute daily increments into field population densities compatible with the experimental plot size via the following expression:

\begin{equation}\label{equation7}
N(t) = \frac{2 \times D(t)}{S}
\end{equation}
where $N(t)$ represents the total adult emergence density at time $t$ (individuals/m$^2$), and $S = 10,000$ m$^2$ signifies the total area of the fixed experimental plot. 
This transformation ensures that the observational data are strictly aligned with the physical dimensions and scale of the stage-structured model (Figure \ref{Figure4}(c)). 

Finally, to mitigate the impact of high-frequency environmental noise which can severely hinder the convergence of parameter inference, we applied a Savitzky-Golay filter \citep{savitzky1964smoothing} to the density time series (Figure \ref{Figure4}(d)).
The filter performs a sliding-window convolution based on a local polynomial least-squares fit. 
The smoothed density values $\hat{N}(t)$ are calculated as:
\begin{equation}\label{equation8}
\hat{N}(t) = \sum_{j=-m}^{m} c_j N(t+j)
\end{equation}
where $c_j$ represents the convolution coefficients determined by a third-order polynomial fit within the sliding window.  
The resulting smoothed target density curve serves as the empirical benchmark for the subsequent estimation of unknown biological parameters within the inversion model.

\begin{figure}[t]
\centering
\includegraphics[width=\linewidth]{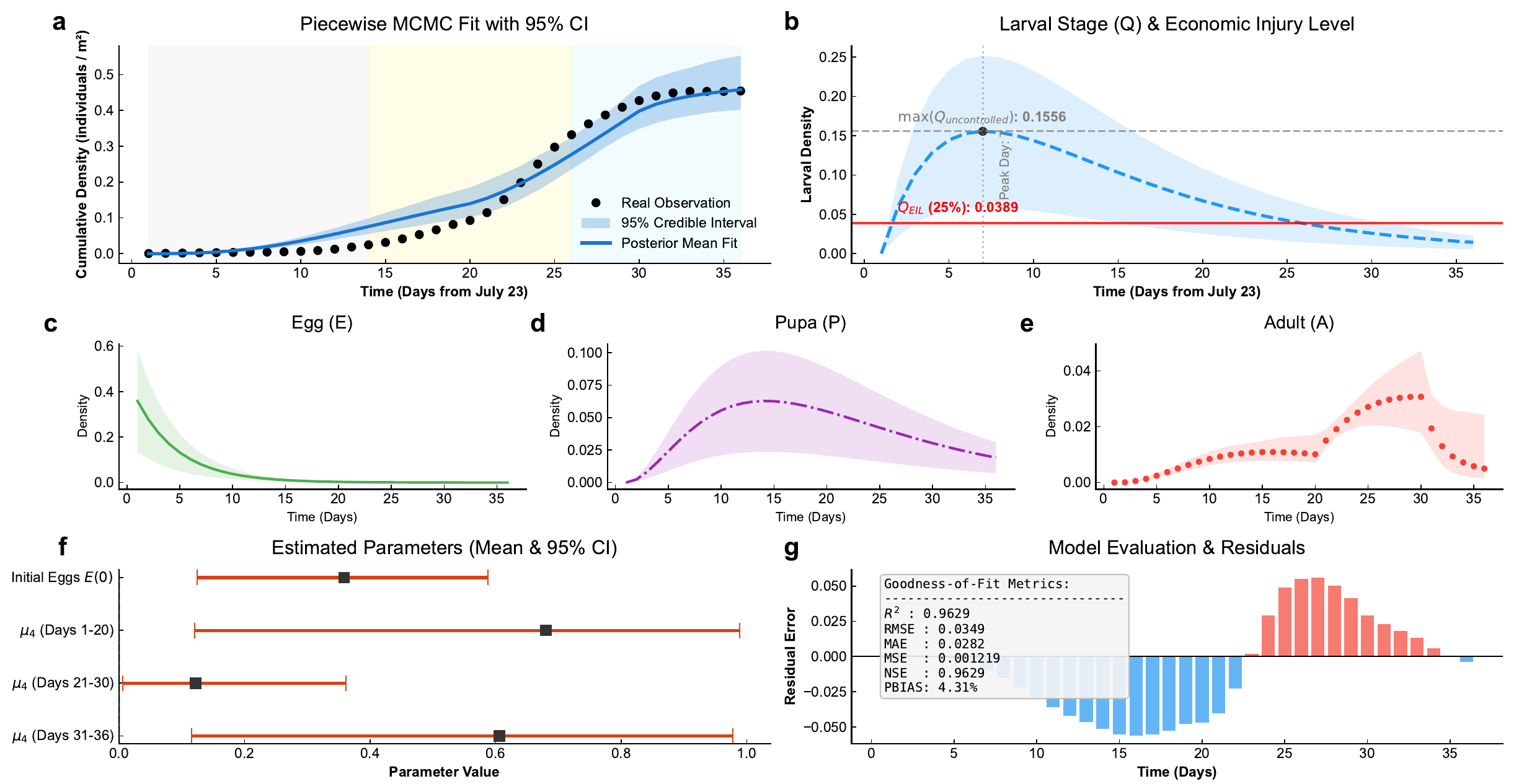}
\caption{\textbf{Comprehensive parameter inference, stage-structured dynamics, and Economic Injury Level (EIL) evaluation dashboard for the piecewise EQPAL system.} 
    (\textbf{a}) Posterior predictive check of the cumulative adult emergence. 
    Black dots denote field observations, the solid blue line represents the posterior mean trajectory, and the shaded area indicates the 95\% credible interval (CI) capturing observational uncertainty. 
    (\textbf{b}) Simulated dynamics of the larval ($Q$) and the establishment of the EIL. 
    The dashed gray line denotes the maximum uncontrolled larval density ($\max(Q_{uncontrolled})$). 
    The solid red line represents the absolute $Q_{EIL}$ threshold, strictly set at 25\% of the uncontrolled peak to align with the macroscopic agricultural policy target of restricting crop yield loss to 5\%. 
    (\textbf{c}-\textbf{e}) Sequential population dynamics of the Egg ($E$), Pupa ($P$), and Adult ($A$). 
    (\textbf{f}) Forest plot of the estimated key parameters. 
    The initial egg density ($E_0$) and time-varying adult mortality rates ($\mu_4$) are displayed with their posterior means (squares) and 95\% CIs (error bars). 
    (\textbf{g}) Model performance metrics and residual distribution. 
    The bar chart illustrates the daily residuals, complemented by goodness-of-fit metrics ($R^2$, RMSE, MAE, and MSE) that corroborate the robustness of the MCMC inversion.}
    \label{Figure5}
\end{figure}   
\unskip

\subsubsection{MCMC method fitting}

To calibrate the dynamic model, we employed the Markov Chain Monte Carlo (MCMC) method to fit the preprocessed adult emergence data \citep{metropolis1953equation,hastings1970monte}.
The posterior prediction test (Figure \ref{Figure5}(a)) shows that the field observation results are in high agreement with the simulated trajectory, and most of the data are within the 95\% confidence interval (CI).
Given the transient environmental factors during the crop season, we estimated the adult natural mortality rate $(\mu_4)$ piecewise across three temporal windows (Days 1-20, 21-30, and 31-36), alongside the initial egg density $E(0)$. 
The posterior means and 95\% CIs of these estimated parameters are summarized in Figure \ref{Figure5}(f).
The robustness of the parameter inversion was verified through residual analysis and goodness-of-fit evaluation (Figure \ref{Figure5}(g)). 
The model achieved a high coefficient of determination ($R^2 = 0.9599$) and low error metrics (RMSE = 0.0363, MAE = 0.0297), indicating that the parameterized EQPAL system accurately captures the authentic outbreak dynamics of soybean pod borer.

Utilizing the fitted parameters, we reconstructed the continuous population dynamics for all biological stages (Figures \ref{Figure5}(b-e)). 
The chronological succession of density peaks across eggs (E), pupae (P), and adults (A) mathematically validates the developmental time delays inherent in the holometabolous life cycle. 
Crucially, as the larvae (Q) constitute the sole destructive stage boring into soybean pods, we quantified their uncontrolled outbreak peak at $\max(Q_{uncontrolled}) = 0.1556$ individuals/$m^2$ (Figure \ref{Figure5}(b)).

In accordance with official agricultural management guidelines which aim to restrict crop yield loss to 5\%, we strictly established the Economic Injury Level $(Q_{EIL})$ at 25\% of the maximum uncontrolled larval density \citep{shaanxi2025soybean}.
 target threshold for all subsequent biocontrol interventions was set to $Q_{EIL} = 0.0389$ individuals/$m^2$ (Figure \ref{Figure5}(b), solid red line).

\begin{figure}[t]
\centering
\includegraphics[width=0.8\linewidth]{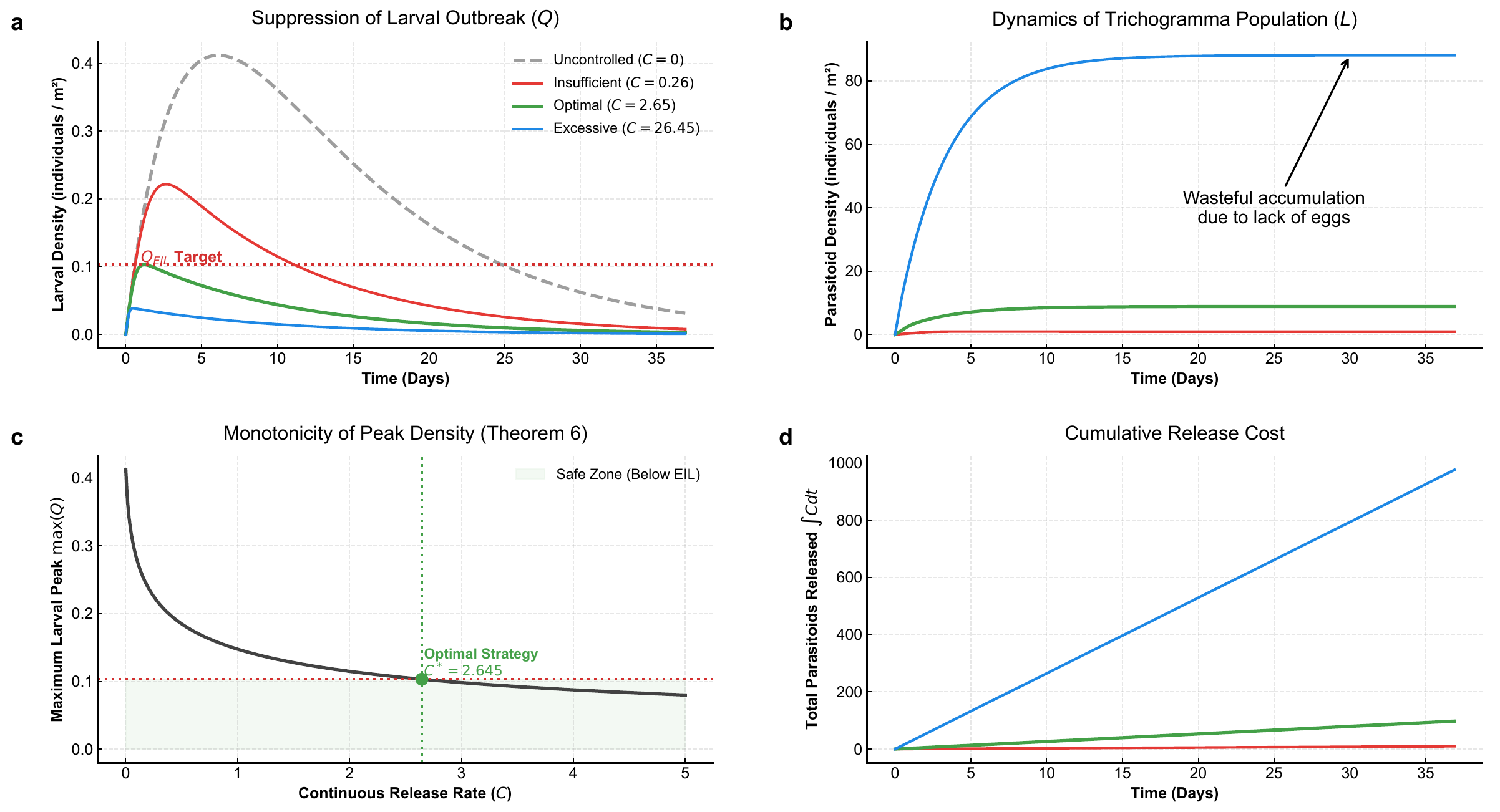}
\caption{
\textbf{Efficacy and economic cost assessment of continuous \textit{Trichogramma} release strategies based on the EQPAL model.} 
    (\textbf{a}) Temporal dynamics of larval density ($Q$) under various continuous release rates ($C$).
    (\textbf{b}) Simulated population trajectories of released\textit{Trichogramma} ($L$).  
    (\textbf{c})  Validation of the monotonic relationship between the release rate ($C$) and the maximum larval peak, identifying the optimal release rate ($C^*$) required to maintain pest populations below the $Q_{EIL}$ threshold. 
    (\textbf{d}) Cumulative economic cost of different continuous release strategies.
}
    \label{Figure6}
\end{figure}

\subsection{Different natural enemy release strategies}
Building upon the parameterized dynamic system, we systematically evaluated continuous and periodic impulsive release strategies of \textit{Trichogramma} wasps to determine the optimal biological control intervention. 

Under the theoretical continuous release scenario (Figure \ref{Figure6}), the maximum larval peak $\max(Q)$ decreases monotonically as the release rate $C$ increases, which aligns with the theoretical prediction of Theorem 6 (Figure \ref{Figure6}(c)). 
However, releasing an excessive number of wasps ($C = 26.45$) yields severely diminishing returns; once host eggs are depleted, the surplus live parasitoids suffer from wasteful accumulation followed by natural mortality and emigration (Figure \ref{Figure6}(b)). To balance pest suppression efficacy with the cumulative cost of rearing and releasing the wasps (Figure \ref{Figure6}(d)), the optimal continuous release rate is strictly determined as $C^* = 2.645$. 
At this continuous rate, the larval outbreak is precisely constrained to the $Q_{EIL}$ target (0.0389 individuals/$m^2$).

Given the logistical constraints and high labor or drone-flight costs associated with daily continuous releases in large-scale agriculture, we further evaluated the practical periodic impulsive release strategy (Figure \ref{Figure7}). 
To isolate the impact of release frequency, the equivalent total number of parasitoids was conserved across different intervals ($T$) by setting the impulsive release volume to $\delta = C^* \times T$. 
The simulations expose a critical ecological vulnerability driven by the release interval. When the interval is extended to $T = 7$ days ($\delta = 18.5$), the short natural lifespan of \textit{Trichogramma} ($\mu_L = 0.25$, yielding an average lifespan of approximately 4 days) creates a parasitoid-free window between consecutive field releases (Figure \ref{Figure7}(b)). 
This temporal gap allows pest eggs to escape parasitism, causing the larval peak to significantly breach the $Q_{EIL}$ threshold (Figure \ref{Figure7}(a)). 
Conversely, while a short interval of $T = 3$ days ($\delta = 7.9$) provides robust pest suppression, it substantially multiplies the frequency of field release operations. 
Consequently, a release interval of $T = 5$ days ($\delta = 13.2$ wasps per m² per release) emerges as the optimal practical strategy. 
As demonstrated in Figure \ref{Figure7}(c), this specific frequency perfectly maintains the pest population within the safe zone while minimizing the operational frequency, offering a robust quantitative guideline for the biological field management of the soybean pod borer.

\section{Discussion}
The development of effective, ecologically sustainable pest management strategies requires a rigorous understanding of the complex population dynamics between target pests and their natural enemies \citep{fischbein2022population}. 
In this study, we developed a novel stage-structured impulsive differential equation model to quantify the interaction between the soybean pod borer and its obligate egg parasitoid, \textit{Trichogramma} spp.. 
By integrating field monitoring data with Markov Chain Monte Carlo (MCMC) parameter estimation, we bridged the gap between theoretical mathematical ecology and applied agricultural entomology, providing a robust quantitative framework for optimizing biological control interventions.  
A primary methodological strength of our approach lies in the explicit stage-structured formulation coupled with stage-specific density dependence. 
Traditional population models often utilize simple Logistic \citep{pearl1920rate} or Lotka-Volterra \citep{volterra1928variations} models that treat pest populations as a single, homogenous compartment. 
However, holometabolous insects exhibit distinct ecological roles across their developmental stages. 
In our model, the integration of a quadratic density-dependent mortality term ($-\eta Q^2$) specifically within the larval compartment accurately captures the intense intraspecific competition for limited spatial and nutritional resources within soybean pods. 
This refinement prevents the unrealistic exponential population explosions seen in simpler models and allows for a highly precise establishment of the Economic Injury Level ($Q_{EIL}$) based solely on the destructive larval stage rather than the harmless adult stage.  

Furthermore, mathematical models are frequently criticized for relying on arbitrary or literature-derived data values that may not reflect local ecological contexts. 
To overcome this, our study was grounded in high-resolution field data collected from Changchun, China. 
The application of the MCMC algorithm enabled the precise inversion of key biological parameters. 
Notably, the piecewise estimation of the adult natural mortality rate ($\mu_4$) provided a mechanistic reflection of the transient environmental pressures and physiological senescence occurring during the crop season. 
The resulting fit metrics ($R^2 = 0.9599$) confirm that our parameterized system reliably mirrors authentic field outbreak dynamics.  

\begin{figure}[t]
\centering
\includegraphics[width=0.8\linewidth]{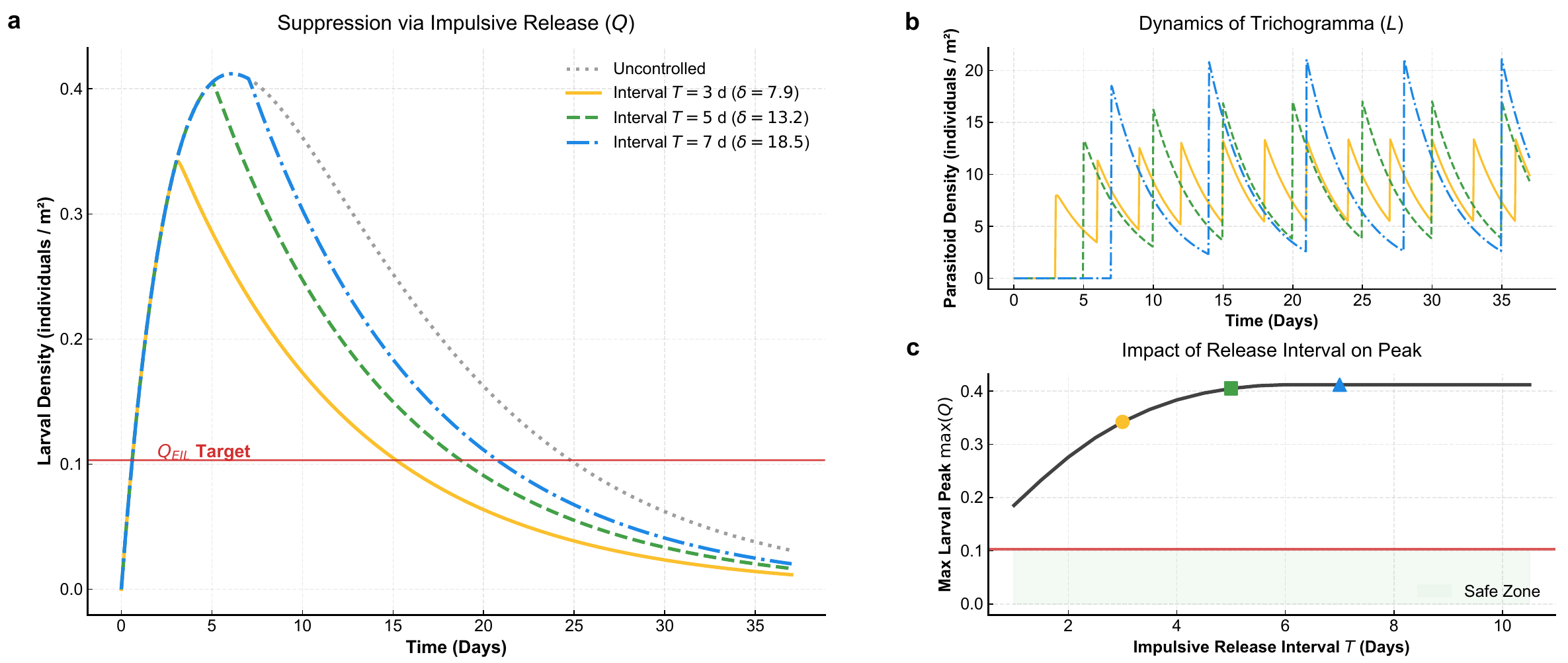}
\caption{
\textbf{Evaluation of periodic impulsive \textit{Trichogramma} release strategies and release interval optimization based on the EQPAL model.} 
    (\textbf{a}) Suppression of larval outbreak dynamics ($Q$) under periodic impulsive releases with varying intervals ($T = 3, 5,$ and $7$ days). 
    The equivalent total release dosage is maintained across different strategies to evaluate the sole impact of release frequency. 
    Distinct line styles and colors differentiate the release intervals, illustrating their respective capacities to constrain the larval population below the strict Economic Injury Level ($Q_{EIL}$ target, solid red line).
    (\textbf{b})  Impulsive population dynamics of released $L$. 
    (\textbf{c}) Impact of the impulsive release interval ($T$) on the maximum larval peak $\max(Q)$.
}
    \label{Figure7}
\end{figure}

The evaluation of continuous release strategies provided profound theoretical insights into the efficiency of biological control. 
While our mathematical proof guarantees a monotonic decrease in the maximum larval peak with increasing release rates, our numerical simulations exposed the economic and ecological paradox of excessive intervention. 
Specifically, when the release rate was excessively high ($C = 26.45$), the rapid depletion of host eggs led to a wasteful accumulation of surplus live parasitoids, which subsequently succumbed to natural mortality and emigration. 
This finding highlights the fallacy of "more is better" in biological control \citep{kriticos2003roles} and underscores the necessity of strictly defining an optimal release threshold ($C^* = 2.645$) to perfectly balance pest suppression with economic viability.  

Translating continuous theoretical models into field reality requires accounting for discrete operations, as daily manual or mechanical release of live insects is economically prohibitive. 
Our analysis of the periodic impulsive release strategy uncovered a critical ecological vulnerability governed by the release interval. Because \textit{Trichogramma} possesses a relatively short natural lifespan, extending the release interval to $T = 7$ days created a "parasitoid-free window". 
During this temporal gap, pest eggs successfully escaped parasitism, causing the destructive larval population to rebound and breach the safety threshold. 
This dynamic is mathematically formalized by our impulsive threshold condition ($\mathcal{R}_{imp}$), proving that in augmentative biological control, the temporal frequency of release is as critical as the total volume of parasitoids deployed.  

From a practical agricultural management perspective, our optimization framework identifies a release interval of $T = 5$ days ($\delta = 13.2$) as the most viable strategy. 
This specific regimen effectively suppresses the pest population within the safe zone while minimizing the operational footprint, thereby offering a direct, actionable guideline for modern Integrated Pest Management (IPM) systems \citep{pradhan2025integrated}.  
Despite its robust findings, this study has certain limitations that pave the way for future research. First, the current ODE-based framework assumes spatial homogeneity, ignoring the dispersal and diffusion dynamics of both the moths and the parasitoids across fragmented agricultural landscapes. 
Future models could incorporate reaction-diffusion equations (PDEs) to evaluate spatially targeted release strategies \citep{ford2015derivation}. 
Second, abiotic environmental covariates such as temperature, precipitation, and wind speed, which significantly influence insect phenology and flight behavior \citep{vebrova2018seasonality}, were implicitly captured through piecewise parameter fitting rather than explicit mechanistic equations. 
Integrating climate-driven thermodynamic sub-models \citep{parasar2025explainable} could further enhance the predictive power of the system under climate change scenarios. 
Finally, expanding the model to consider the sublethal effects of biopesticides combined with \textit{Trichogramma} releases could provide a more comprehensive view of integrated management strategies in commercial soybean production.

\section{Conclusion}
This study establishes a robust, data-driven mathematical framework for optimizing \textit{Trichogramma} release strategies against the soybean pod borer. 
By developing a novel stage-structured dynamic model that explicitly couples larval density dependence with obligate egg parasitism, and rigorously parameterizing it using field monitoring data via the MCMC method, we accurately quantified the pest's biological trajectory and established a precise larval Economic Injury Level ($Q_{EIL}$). 
Our theoretical and numerical analyses reveal that while a continuous release rate of $C^* = 2.645$ optimally suppresses the pest theoretically, it is operationally demanding. 
Crucially, the evaluation of periodic impulsive releases demonstrates that the temporal frequency of interventions is strictly bottlenecked by the short natural lifespan of the parasitoids. 
Extending the release interval to 7 days inevitably leads to control failure due to the emergence of a parasitoid-free window. 
We conclude that a 5-day impulsive release interval provides the optimal biological control strategy, ensuring strict pest suppression within the safe zone while minimizing field operational costs. 
Ultimately, this quantitative decision-making tool bridges the gap between theoretical population dynamics and applied agricultural management, providing scientifically viable guidelines for the green and precise control of soybean pests.

\printcredits

\section*{Data availability statement}

All pest data and Python source code in the article can be obtained from the following repositories: \url{https://github.com/jluLWX/Data-and-code-of-soybean-aphid} (accessed on 20 July 2026).

\section*{Declaration of competing interest}

The authors declare that we have no known competing financial interests or personal relationships that could have appeared to influence the work reported in this paper.

\section*{Acknowledgements}

This research was funded by the National Natural Science Foundation of China (Grant No. 12301627 for Suli Liu); the Science and Technology Research Projects of the Education Office of Jilin Province, China (Grant No. JJKH20250046KJ for Suli Liu); 
the National Key Research and Development Program of China (Grant No. 2023YFD1401000 for Yu Gao); the Earmarked Fund for China Agriculture Research System of MOF and MARA (Grant No. CARS04 for Yu Gao).

\appendix
\section[\appendixname~\thesection]{Mathematical proof process of  Without Artificial Release}\label{appA}
\subsection[\appendixname~\thesubsection]{Theorem 1}\label{app1}
\begin{theorem}\label{th1}
For any given initial condition $E_0, Q_0, P_0, A_0, L_0 \in \mathbb{R}^5_+$ the solutions $E(t), Q(t)$,\\
$ P(t), A(t), L(t)$ of the system are positive or zero for all $t \ge 0$.
\end{theorem}
\begin{pf} Based on the existence and uniqueness theorem of ordinary differential equations we examine the vector field direction on each boundary hyperplane of the region $\mathbb{R}^5_+$. Setting each respective state variable to zero yields the following derivatives.

\begin{align*}
\frac{dE}{dt} \bigg|_{E=0} &= bA \ge 0,\quad \frac{dQ}{dt} \bigg|_{Q=0} = \alpha_1 E \ge 0, \\
\frac{dP}{dt} \bigg|_{P=0} &= \alpha_2 Q \ge 0,\quad \frac{dA}{dt} \bigg|_{A=0} = \alpha_3 P \ge 0.
\end{align*}

Integrating the final equation of the system directly provides the analytical form for the natural enemy population.

$$
L(t) = L_0 \exp \int_{0}^{t} \left[ c\beta E(s) - \mu_L - \omega \right] ds
$$

The exponential function is strictly positive meaning $L(t) \ge 0$ holds globally for all $t \ge 0$. All vector fields on the boundaries point inward or remain tangent. The region $\mathbb{R}^5_+$ is therefore a positively invariant set.
\end{pf}

\subsection[\appendixname~\thesubsection]{Theorem 2}\label{app2}

\begin{theorem}\label{th2}
    There exists a positive constant $M > 0$ such that all solutions of system \eqref{equation3} satisfy $E(t), Q(t), P(t), A(t), L(t) \le M$ for sufficiently large $t$.
\end{theorem}

\begin{pf}
We define a positively weighted total population function
$$V(t) = cE(t) + cQ(t) + k_1 P(t) + k_2 A(t) + L(t)$$
where $c$ is the conversion rate parameter in the system and $k_1, k_2$ are positive constants to be determined. Calculating the derivative of $V(t)$ along the trajectories of system \eqref{equation3} yields
$$
\begin{aligned}
\frac{dV}{dt} =& c \frac{dE}{dt} + c \frac{dQ}{dt} + k_1 \frac{dP}{dt} + k_2 \frac{dA}{dt} + \frac{dL}{dt} \\
=& -c\mu_1 E + [k_1 \alpha_2 - c(\alpha_2 + \mu_2)]Q - c\eta Q^2 \\
&+ [k_2 \alpha_3 - k_1 (\alpha_3 + \mu_3)]P + (cb - k_2 \mu_4)A - (\mu_L + \omega)L
\end{aligned}
$$
We choose sufficiently large constants $k_2$ and $k_1$ such that $cb - k_2 \mu_4 < 0$ and $k_2 \alpha_3 - k_1 (\alpha_3 + \mu_3) < 0$. Let $cb - k_2 \mu_4 = -\delta_A$ and $k_2 \alpha_3 - k_1 (\alpha_3 + \mu_3) = -\delta_P$ where $\delta_A, \delta_P > 0$. We then select a sufficiently small positive constant $\delta$ satisfying
$$ \delta \le \min \left\{ \mu_1, \frac{\delta_P}{k_1}, \frac{\delta_A}{k_2}, \mu_L + \omega \right\} $$
Adding $\delta V$ to both sides of the derivative equation gives
$$ \frac{dV}{dt} + \delta V \le M_1 Q - c\eta Q^2 $$
where $M_1 = \delta c + k_1 \alpha_2 - c(\alpha_2 + \mu_2)$. The right side of the inequality is a quadratic polynomial in $Q$ opening downwards. It necessarily possesses a maximum value $M^* = \frac{M_1^2}{4c\eta}$ for all $Q \in \mathbb{R}$. We obtain the differential inequality
$$ \frac{dV}{dt} + \delta V \le M^* $$
Applying the standard comparison theorem yields
$$ \limsup_{t \to \infty} V(t) \le \frac{M^*}{\delta} $$
Since $V(t)$ is a linear combination of state variables with strictly positive coefficients and all variables are non-negative from Theorem 1 it follows that there exists a positive constant $M$ such that all state variables are ultimately bounded.
\end{pf}

\subsection[\appendixname~\thesubsection]{Lemma 3}\label{app3}

\begin{lemma}\label{le3}
    Define the constant $K_1 = \frac{(\alpha_3 + \mu_3)\mu_4}{\alpha_2 \alpha_3}$ and the natural survival threshold $\Lambda = \frac{\alpha_1 b}{\alpha_1 + \mu_1} - (\alpha_2 + \mu_2)K_1$. If $\Lambda > 0$, system \eqref{equation3} has a unique positive parasitoid-free equilibrium $E_0(E_0, Q_0, P_0, A_0, 0)$.
\end{lemma}

\begin{pf}
To find the parasitoid-free equilibrium, we set the derivatives of all state variables in system \eqref{equation3} to zero and let $L = 0$. This yields the following algebraic equations:
\begin{align*}
bA - (\alpha_1 + \mu_1)E &= 0,\quad \alpha_1 E - (\alpha_2 + \mu_2)Q - \eta Q^2 = 0, \\
\alpha_2 Q - (\alpha_3 + \mu_3)P &= 0,\quad \alpha_3 P - \mu_4 A = 0.
\end{align*}
Solving this system from bottom to top, we can express $P, Q,$ and $E$ in terms of $A$:
$$ P = \frac{\mu_4}{\alpha_3} A ,\quad Q = \frac{\alpha_3 + \mu_3}{\alpha_2},\quad
P = \frac{(\alpha_3 + \mu_3)\mu_4}{\alpha_2 \alpha_3} A = K_1 A, \quad E = \frac{b}{\alpha_1 + \mu_1} A $$
Substituting the expressions for $E$ and $Q$ into the second equation gives:
$$ \frac{\alpha_1 b}{\alpha_1 + \mu_1} A - (\alpha_2 + \mu_2)K_1 A - \eta K_1^2 A^2 = 0 $$
Since we are looking for a positive equilibrium $A \neq 0$. Factoring out $A$ yields a linear equation:
$$ \frac{\alpha_1 b}{\alpha_1 + \mu_1} - (\alpha_2 + \mu_2)K_1 - \eta K_1^2 A = 0 $$
Solving for $A$ provides the equilibrium density $A_0$:
$$ A_0 = \frac{1}{\eta K_1^2} \left[ \frac{\alpha_1 b}{\alpha_1 + \mu_1} - (\alpha_2 + \mu_2)K_1 \right] = \frac{\Lambda}{\eta K_1^2} $$
It is clear that $A_0 > 0$ if and only if $\Lambda > 0$. Substituting $A_0$ back into the previous relations uniquely determines strictly positive values for $P_0, Q_0,$ and $E_0$. Thus a unique positive parasitoid-free equilibrium $E_0$ exists when $\Lambda > 0$.
\end{pf}

\subsection[\appendixname~\thesubsection]{Definition 1}\label{app4}

\begin{definition}\label{de1}
    The basic reproduction number of system \eqref{equation3}, denoted by $\mathcal{R}_0$, is defined as the average number of offspring produced by a single \textit{Trichogramma} individual during its entire lifespan in a pest-outbreak environment ($E = E_0$):
    $$
        \mathcal{R}_0 = \frac{c\beta E_0}{\mu_L + \omega}
    $$
\end{definition}

\begin{pf}
To derive the basic reproduction number $\mathcal{R}_0$, we employ the next-generation matrix method \citep{vandendriessche2002reproduction} by linearizing system \eqref{equation3} at the parasitoid-free equilibrium $E_0(E_0, Q_0, P_0, A_0, 0)$. Since the natural enemy population $L$ is the only compartment involved in the transmission of parasitism, we focus on the dynamics of $L$:
$$ \frac{dL}{dt} = c\beta EL - (\mu_L + \omega)L $$
At the equilibrium $E_0$, the pest egg density is stabilized at $E = E_0$. The linearized equation for $L$ is given by:
$$ \frac{dL}{dt} = \left[ c\beta E_0 - (\mu_L + \omega) \right] L $$
Following the next-generation matrix theory, we decompose the rate of change of $L$ into the production of new individuals $\mathcal{F}$ and the net rate of out-flow $\mathcal{V}$:
$$ F = \frac{\partial (c\beta E L)}{\partial L} \bigg|_{E=E_0} = c\beta E_0, \quad V = \frac{\partial [(\mu_L + \omega)L]}{\partial L} \bigg|_{E=E_0} = \mu_L + \omega $$
The basic reproduction number $\mathcal{R}_0$ is the spectral radius of the next-generation matrix $FV^{-1}$. In this scalar case, it simplifies to:
$$ \mathcal{R}_0 = FV^{-1} = \frac{c\beta E_0}{\mu_L + \omega} $$
The biological interpretation of $\mathcal{R}_0$ is straightforward: $(\mu_L + \omega)^{-1}$ represents the expected lifespan of a wild \textit{Trichogramma}, and $c\beta E_0$ is the number of offspring produced per unit time. 
Their product yields the total reproduction potential over the parasitoid's life.
\end{pf}

\subsection[\appendixname~\thesubsection]{Theorem 4}\label{app5}

\begin{theorem}\label{th4}
    The parasitoid-free equilibrium $E_0$ of system \eqref{equation3} is locally asymptotically stable if $\mathcal{R}_0 < 1$, and unstable if $\mathcal{R}_0 > 1$.
\end{theorem}

\begin{pf}
The local stability of $E_0(E_0, Q_0, P_0, A_0, 0)$ is determined by the Jacobian matrix of system \eqref{equation3} evaluated at this point. The Jacobian matrix $J(E_0)$ takes the following block-triangular form:
$$
J(E_0) = 
\begin{pmatrix}
M_{4 \times 4} & W_{4 \times 1} \\
\mathbf{0}_{1 \times 4} & c\beta E_0 - (\mu_L + \omega)
\end{pmatrix}
$$
where $M_{4 \times 4}$ represents the Jacobian of the pest sub-system:
$$
M = \begin{pmatrix}
-(\alpha_1 + \mu_1) & 0 & 0 & b \\
\alpha_1 & -(\alpha_2 + \mu_2) - 2\eta Q_0 & 0 & 0 \\
0 & \alpha_2 & -(\alpha_3 + \mu_3) & 0 \\
0 & 0 & \alpha_3 & -\mu_4
\end{pmatrix}
$$
The eigenvalues of $J(E_0)$ consist of the eigenvalues of $M$ and the scalar $\lambda_5 = c\beta E_0 - (\mu_L + \omega)$. 
We can rewrite $\lambda_5$ in terms of the basic reproduction number:
$$ \lambda_5 = (\mu_L + \omega) \left( \frac{c\beta E_0}{\mu_L + \omega} - 1 \right) = (\mu_L + \omega)(\mathcal{R}_0 - 1) $$
Clearly, $\lambda_5 < 0$ if $\mathcal{R}_0 < 1$, and $\lambda_5 > 0$ if $\mathcal{R}_0 > 1$. 
For the sub-matrix $M$, all its diagonal elements are negative. Since $E_0$ is a positive equilibrium existing under the condition $\Lambda > 0$, it can be shown via the Routh-Hurwitz criterion \citep{gantmacher1959applications} that all eigenvalues of $M$ have negative real parts. 
Thus, $E_0$ is locally asymptotically stable when $\mathcal{R}_0 < 1$ and unstable when $\mathcal{R}_0 > 1$.
\end{pf}

\subsection[\appendixname~\thesubsection]{Lemma 5}\label{app6}

\begin{lemma}\label{le5}
    System \eqref{equation3} has a unique positive coexistence equilibrium $E^*(E^*, Q^*, P^*, A^*, L^*)$ if and only if $\mathcal{R}_0 > 1$.
\end{lemma}

\begin{pf}
The coexistence equilibrium satisfies the algebraic equations of system \eqref{equation3} with all populations being strictly positive. From $\frac{dL}{dt} = 0$ and $L \neq 0$, we immediately obtain the equilibrium density of eggs:
$$ E^* = \frac{\mu_L + \omega}{c\beta} = \frac{E_0}{\mathcal{R}_0} $$
Substituting $E^*$ into the equation $\frac{dQ}{dt} = 0$ yields a quadratic equation for $Q^*$:
$$ \eta (Q^*)^2 + (\alpha_2 + \mu_2)Q^* - \alpha_1 E^* = 0 $$
Since all parameters are positive and $E^* > 0$, this quadratic equation has exactly one positive root $Q^*$. Subsequently, $P^*$ and $A^*$ are uniquely determined by:
$$ P^* = \frac{\alpha_2}{\alpha_3 + \mu_3} Q^*, \quad A^* = \frac{\alpha_3 P^*}{\mu_4} = \frac{\alpha_2 \alpha_3}{\mu_4 (\alpha_3 + \mu_3)} Q^* $$
Finally, from $\frac{dE}{dt} = 0$, the equilibrium density of the natural enemy $L^*$ is expressed as:
$$ L^* = \frac{b A^* - (\alpha_1 + \mu_1)E^*}{\beta E^*} $$
For $L^* > 0$, we require $b A^* > (\alpha_1 + \mu_1)E^*$. Recalling the definition of $E_0$ and $E^* = E_0 / \mathcal{R}_0$, this inequality holds if and only if $E^* < E_0$, which is equivalent to $\mathcal{R}_0 > 1$. Therefore, a unique positive coexistence equilibrium $E^*$ exists if and only if $\mathcal{R}_0 > 1$.
\end{pf}

\subsection[\appendixname~\thesubsection]{Theorem 6}\label{app7}

\begin{theorem}\label{th6}
    Suppose $\mathcal{R}_0 > 1$. The unique positive coexistence equilibrium $E^*(E^*, Q^*, P^*, A^*, L^*)$ of system \eqref{equation3} is globally asymptotically stable in the interior of the positive orthant $\mathbb{R}^5_+$.
\end{theorem}

\begin{pf}
Consider the following Lyapunov function:
$$ V = cE^* G\left(\frac{E}{E^*}\right) + cQ^* G\left(\frac{Q}{Q^*}\right) + k_1 P^* G\left(\frac{P}{P^*}\right) + k_2 A^* G\left(\frac{A}{A^*}\right) + L^* G\left(\frac{L}{L^*}\right) $$
where $G(x) = x - 1 - \ln x$. Differentiating $V$ along the trajectories of system \eqref{equation3} and substituting the equilibrium conditions, we obtain:
\begin{align*}
\dot{V} =& c\left(1 - \frac{E^*}{E}\right) [bA - (\alpha_1 + \mu_1)E - \beta EL]+ c\left(1 - \frac{Q^*}{Q}\right) [\alpha_1 E - (\alpha_2 + \mu_2)Q - \eta Q^2 + c\beta EL] \\
&+ k_1 \left(1 - \frac{P^*}{P}\right) [\alpha_2 Q - (\alpha_3 + \mu_3)P]+ k_2 \left(1 - \frac{A^*}{A}\right) [\alpha_3 P - \mu_4 A] \\
&+ \left(1 - \frac{L^*}{L}\right) [c\beta EL - (\mu_L + \omega)L]
\end{align*}
Using the identities $(\alpha_1 + \mu_1)E^* = bA^* - \beta E^* L^*$, $(\alpha_2 + \mu_2)Q^* = \alpha_1 E^* - \eta (Q^*)^2 + c\beta E^* L^*$, and so on, we can rearrange $\dot{V}$ as follows:
\begin{align*}
\dot{V} = & -c\eta \frac{(Q-Q^*)^2}{Q} Q^*+ c(\alpha_1 + \mu_1)E^* \left( 2 - \frac{E}{E^*} - \frac{E^*}{E} \right) + c\alpha_1 E^* \left( 3 - \frac{E^*}{E} - \frac{E Q^*}{E^* Q} - \frac{Q}{Q^*} \right) \\
& + k_1 \alpha_2 Q^* \left( 3 - \frac{Q^*}{Q} - \frac{Q P^*}{Q^* P} - \frac{P}{P^*} \right)+ k_2 \alpha_3 P^* \left( 3 - \frac{P^*}{P} - \frac{P A^*}{P^* A} - \frac{A}{A^*} \right) \\
& + c\beta E^* L^* \left( 4 - \frac{E^*}{E} - \frac{A L E^*}{A^* L^* E} - \frac{A^*}{A} - \frac{L}{L^*} \right)
\end{align*}
According to the arithmetic-geometric mean inequality, the terms in each parenthesis satisfy $n - \sum x_i \le 0$ because the product of the ratios in each set equals 1. Specifically, $(2 - E/E^* - E^*/E) \le 0$ and $(3 - E^*/E - EQ^*/E^*Q - Q/Q^*) \le 0$, with equality holding if and only if $E=E^*, Q=Q^*, P=P^*, A=A^*, L=L^*$. The term $-c\eta (Q-Q^*)^2 Q^*/Q$ is also strictly negative for $Q \neq Q^*$. Therefore, $\dot{V} \le 0$ and $\dot{V} = 0$ only at the equilibrium $E^*$. By LaSalle's Invariance Principle, $E^*$ is globally asymptotically stable in the interior of $\mathbb{R}^5_+$.
\end{pf}

\section[\appendixname~\thesection]{Mathematical proof process of Continuous Constant Release Strategy}\label{appB}
\subsection[\appendixname~\thesubsection]{Lemma 7}\label{app8}

\begin{lemma}\label{le7}
    For any constant artificial release rate $C > 0$, system \eqref{equation4} does not admit a parasitoid-free equilibrium.
\end{lemma}

\begin{pf}
Let us assume, for the sake of contradiction, that there exists a parasitoid-free equilibrium $\tilde{E}_0$ for system \eqref{equation4}. By definition, the parasitoid population density at this steady state must satisfy $L = 0$. 
Substituting $L = 0$ into the fifth equation of system \eqref{equation4}, we obtain the rate of change for the parasitoid population:
$$
    \frac{dL}{dt}\bigg|_{L=0} = c\beta E(0) - (\mu_L + \omega)(0) + C = C
$$
Given the assumption that the artificial release rate is a positive constant ($C > 0$), it follows that $\dot{L} = C > 0$ whenever $L = 0$. This indicates that the vector field on the boundary $L = 0$ always points into the interior of the positive orthant $\mathbb{R}^5_+$, contradicting the requirement for an equilibrium point ($\dot{L} = 0$). Thus, the parasitoid population cannot vanish, and no parasitoid-free equilibrium exists.
\end{pf}

\subsection[\appendixname~\thesubsection]{Lemma 8}\label{app9}

\begin{lemma}\label{le8}
    For any constant artificial release rate $C > 0$, system \eqref{equation4} possesses a unique positive coexistence equilibrium $E^*_C(E^*_C, Q^*_C, P^*_C, A^*_C, L^*_C)$.
\end{lemma}

\begin{pf}
To find the coexistence equilibrium, we set the derivatives in system \eqref{equation4} to zero. 
From the equations for $\dot{Q}, \dot{P},$ and $\dot{A}$, we can express $Q, P,$ and $A$ as functions of the egg density $E$:
$$
    Q(E) = \frac{-(\alpha_2 + \mu_2) + \sqrt{(\alpha_2 + \mu_2)^2 + 4\eta\alpha_1 E}}{2\eta}, \quad P(E) = \frac{\alpha_2 Q(E)}{\alpha_3 + \mu_3}, \quad A(E) = \frac{\alpha_3 P(E)}{\mu_4}
$$
It is clear that $Q(E), P(E),$ and $A(E)$ are strictly increasing functions of $E$ for $E > 0$. 
From $\dot{E} = 0$, we solve for the parasitoid density $L$ as a function of $E$:
$$
    L(E) = \frac{bA(E) - (\alpha_1 + \mu_1)E}{\beta E}
$$
Substituting $L(E)$ into the equation $\dot{L} = 0$, we define the transcendental function $f(E)$:
$$
    f(E) = [c\beta E - (\mu_L + \omega)] L(E) + C = 0
$$
As $E \to 0^+$, $L(E) \to \infty$ and $c\beta E - (\mu_L + \omega) < 0$, which implies $f(E) \to -\infty$. As $E$ approaches the parasitoid-free equilibrium density $E_0$ (the density where $L(E_0)=0$), we have $f(E_0) = C > 0$. Since $f(E)$ is a continuous and monotonically increasing function on the interval $(0, E_0)$, there exists a unique root $E^*_C \in (0, E_0)$ such that $f(E^*_C) = 0$. The other steady-state components $Q^*_C, P^*_C, A^*_C,$ and $L^*_C$ are uniquely determined and strictly positive.
\end{pf}

\subsection[\appendixname~\thesubsection]{Theorem 9}\label{app10}

\begin{theorem}\label{th9}
    For any constant artificial release rate $C > 0$, the unique positive coexistence equilibrium $E^*_C$ of system \eqref{equation4} is globally asymptotically stable in the interior of $\mathbb{R}^5_+$.
\end{theorem}

\begin{pf}
We employ the same Volterra-type Lyapunov function $V_C$ as defined in Theorem 4. Differentiating $V_C$ along the trajectories of system \eqref{equation4} and substituting the equilibrium condition $C = (\mu_L + \omega)L^*_C - c\beta E^*_C L^*_C$, the final form of the derivative $\dot{V}_C$ is derived as:
\begin{align*}
\dot{V}_C = & -c\eta \frac{(Q-Q^*_C)^2}{Q} Q^*_C + c(\alpha_1 + \mu_1)E^*_C \left( 2 - \frac{E}{E^*_C} - \frac{E^*_C}{E} \right) \ + c\alpha_1 E^*_C \left( 3 - \frac{E^*_C}{E} - \frac{E Q^*_C}{E^*_C Q} - \frac{Q}{Q^*_C} \right) \\
& + k_1 \alpha_2 Q^*_C \left( 3 - \frac{Q^*_C}{Q} - \frac{Q P^*_C}{Q^*_C P} - \frac{P}{P^*_C} \right) + k_2 \alpha_3 P^*_C \left( 3 - \frac{P^*_C}{P} - \frac{P A^*_C}{P^*_C A} - \frac{A}{A^*_C} \right) \\
& + c\beta E^*_C L^*_C \left( 4 - \frac{E^*_C}{E} - \frac{A L E^*_C}{A^*_C L^*_C E} - \frac{A^*_C}{A} - \frac{L}{L^*_C} \right) + C \left( 2 - \frac{L}{L^*_C} - \frac{L^*_C}{L} \right)
\end{align*}
According to the arithmetic-geometric mean (AM-GM) inequality, each term in the parentheses is non-positive. Specifically, the new term $C(2 - L/L^*_C - L^*_C/L) \le 0$ for all $L > 0$ and $C > 0$, with equality holding if and only if $L = L^*_C$. 

Therefore, we have $\dot{V}_C \le 0$ for all $(E, Q, P, A, L) \in \mathbb{R}^5_+$, and $\dot{V}_C = 0$ holds only at the coexistence equilibrium $E^*_C$. By LaSalle’s Invariance Principle, the equilibrium $E^*_C$ is globally asymptotically stable.
\end{pf}

\subsection[\appendixname~\thesubsection]{Theorem 10}\label{app11}

\begin{theorem}\label{th10}
    The steady-state pest densities $(E^*_C, Q^*_C, P^*_C, A^*_C)$ are strictly monotonically decreasing functions of the artificial release rate $C$. That is, $\frac{dE^*_C}{dC} < 0$, $\frac{dQ^*_C}{dC} < 0$, $\frac{dP^*_C}{dC} < 0$, and $\frac{dA^*_C}{dC} < 0$ for all $C > 0$.
\end{theorem}

\begin{pf}
Based on the existence proof in Lemma 4, the coexistence equilibrium egg density $E^*_C$ is determined by the transcendental equation $f(E, C) = 0$:
$$
    f(E, C) = [c\beta E - (\mu_L + \omega)] L(E) + C = 0
$$
where $L(E) = \frac{bA(E) - (\alpha_1 + \mu_1)E}{\beta E}$. According to the Implicit Function Theorem, the derivative of $E^*_C$ with respect to $C$ is given by:
$$
    \frac{dE^*_C}{dC} = - \frac{\partial f / \partial C}{\partial f / \partial E}
$$
It is clear that $\frac{\partial f}{\partial C} = 1 > 0$. To determine the sign of $\frac{\partial f}{\partial E}$, we rewrite the equilibrium condition as $H(E) = (\mu_L + \omega - c\beta E) L(E) = C$. Differentiating $H(E)$ with respect to $E$ yields:
$$
    H'(E) = -c\beta L(E) + (\mu_L + \omega - c\beta E) \frac{dL}{dE}
$$
At the coexistence equilibrium with $C > 0$ and $L > 0$, the term $(\mu_L + \omega - c\beta E^*_C)$ must be strictly positive. Furthermore, the derivative of the parasitoid density function is:
$$
    \frac{dL}{dE} = \frac{b}{\beta} \cdot \frac{E A'(E) - A(E)}{E^2}
$$
Due to the density-dependent mortality $-\eta Q^2$ in the larval stage, the adult production function $A(E)$ is a strictly concave function passing through the origin. For any such concave function, the marginal value is less than the average value ($A'(E) < A(E)/E$), implying $\frac{dL}{dE} < 0$. 

Since both terms in $H'(E)$ are negative, we have $H'(E) < 0$, which implies $\frac{\partial f}{\partial E} = -H'(E) > 0$. Therefore:
$$
    \frac{dE^*_C}{dC} = - \frac{1}{\partial f / \partial E} < 0
$$
Given that $Q(E), P(E),$ and $A(E)$ are strictly increasing functions of $E$, it follows from the chain rule that $\frac{dQ^*_C}{dC} < 0$, $\frac{dP^*_C}{dC} < 0$, and $\frac{dA^*_C}{dC} < 0$. This completes the proof that increasing the continuous release rate monotonically suppresses the pest population at all developmental stages.
\end{pf}

\section[\appendixname~\thesection]{Mathematical proof process of  periodic impulsive release strategy}\label{appC}
\subsection[\appendixname~\thesubsection]{Lemma 11}\label{app12}

\begin{lemma}\label{le11}
    System \eqref{equation5} possesses a unique and stable positive impulsive periodic solution $\widetilde{L}(t)$, and for any solution $L(t)$ of the parasitoid-only sub-system, $|L(t) - \widetilde{L}(t)| \to 0$ as $t \to \infty$.
\end{lemma}

\begin{pf}
In the absence of the soybean pod borer population, specifically when $E(t)=Q(t)=P(t)=A(t)=0$, the dynamics of the natural enemy population $L(t)$ are governed by the following linear impulsive differential sub-system:
\begin{align*}
\frac{dL(t)}{dt} &= -(\mu_{L}+\omega)L(t), & t \neq nT, \\
L(t^+) &= L(t) + \delta, & t = nT.
\end{align*}
For any time interval $nT < t \le (n+1)T$, integrating the first equation yields the analytical expression $L(t) = L(nT^+) e^{-(\mu_{L}+\omega)(t-nT)}$. Utilizing the impulsive jump condition at $t=(n+1)T$, we establish a stroboscopic map reflecting the population density immediately after successive release events: $L((n+1)T^+) = L(nT^+) e^{-(\mu_{L}+\omega)T} + \delta$. To find the fixed point of this iterative map, denoted by $L^*$, we solve the algebraic identity $L^* = L^* e^{-(\mu_{L}+\omega)T} + \delta$, which gives $L^* = \frac{\delta}{1 - e^{-(\mu_{L}+\omega)T}}$. 
By substituting this initial value into the analytical solution, we obtain the unique positive impulsive periodic solution:
$$
\widetilde{L}(t) = \frac{\delta e^{-(\mu_{L}+\omega)(t-nT)}}{1 - e^{-(\mu_{L}+\omega)T}}, \quad nT < t \le (n+1)T
$$
The global stability of this periodic solution is confirmed by considering the difference between any arbitrary solution $L(t)$ and $\widetilde{L}(t)$. The evolution of this difference follows the homogeneous linear impulsive system, which converges to zero at an exponential rate determined by $e^{-(\mu_{L}+\omega)T}$. Since the parameters $\mu_L, \omega,$ and $T$ are strictly positive, the magnitude of the difference vanishes as $n \to \infty$, ensuring that the parasitoid population always stabilizes to the periodic oscillation $\widetilde{L}(t)$ regardless of initial conditions \citep{lakshmikantham1989theory}.
\end{pf}

\subsection[\appendixname~\thesubsection]{Theorem 12}\label{app13}

\begin{theorem}\label{th12}
    Define the impulsive threshold as $\mathcal{R}_{imp} = \frac{b \alpha_1 \alpha_2 \alpha_3}{\mu_4 (\alpha_3 + \mu_3) (\alpha_2 + \mu_2) (\alpha_1 + \mu_1 + \beta \langle L \rangle)}$, where $\langle L \rangle = \frac{1}{T} \int_0^T \widetilde{L}(t) dt$ is the average density of the parasitoid population over one period. If $\mathcal{R}_{imp} < 1$, then the pest-extinction periodic solution $(0, 0, 0, 0, \widetilde{L}(t))$ of system \eqref{equation5} is globally asymptotically stable.
\end{theorem}

\begin{pf}
The proof is divided into two stages: local stability analysis via Floquet theory and global attractivity analysis using the impulsive comparison theorem.

First, we examine the local stability of the periodic solution $(0, 0, 0, 0, \widetilde{L}(t))$. The linearized system for the pest compartments $(E, Q, P, A)$ around this solution is decoupled from the parasitoid population $L$. The stability is determined by the eigenvalues (Floquet multipliers) of the monodromy matrix $M(T)$ of the following linear periodic system:
\begin{equation*}
    \begin{pmatrix} \dot{E} \\ \dot{Q} \\ \dot{P} \\ \dot{A} \end{pmatrix} = 
    \begin{pmatrix} 
    -(\alpha_1 + \mu_1 + \beta \widetilde{L}(t)) & 0 & 0 & b \\
    \alpha_1 & -(\alpha_2 + \mu_2) & 0 & 0 \\
    0 & \alpha_2 & -(\alpha_3 + \mu_3) & 0 \\
    0 & 0 & \alpha_3 & -\mu_4
    \end{pmatrix} 
    \begin{pmatrix} E \\ Q \\ P \\ A \end{pmatrix}
\end{equation*}
According to Floquet theory, the pest-extinction periodic solution is locally stable if the spectral radius of the monodromy matrix is less than unity, which is equivalent to the condition $\mathcal{R}_{imp} < 1$. Furthermore, the Floquet multiplier for the parasitoid sub-system is $e^{-(\mu_L + \omega)T} < 1$, ensuring the stability of the natural enemy's periodic oscillation.

Second, we prove global attractivity. Since $\mathcal{R}_{imp} < 1$, for a sufficiently small $\epsilon > 0$, we have $L(t) > \widetilde{L}(t) - \epsilon$ for $t > t_0$ based on Lemma 5. Substituting this into the first equation of system \eqref{equation5} and ignoring the non-positive density-dependent term $-\eta Q^2$ in the second equation, we construct a linear comparison system:
\begin{align*}
\dot{E}_s &= bA_s - (\alpha_1 + \mu_1 + \beta(\widetilde{L}(t) - \epsilon))E_s, \\
\dot{Q}_s &= \alpha_1 E_s - (\alpha_2 + \mu_2)Q_s, \\
\dot{P}_s &= \alpha_2 Q_s - (\alpha_3 + \mu_3)P_s, \\
\dot{A}_s &= \alpha_3 P_s - \mu_4 A_s.
\end{align*}
As the average growth rate of this linear periodic comparison system is negative when $\mathcal{R}_{imp} < 1$, all its solutions satisfy $\lim_{t \to \infty} (E_s, Q_s, P_s, A_s) = (0, 0, 0, 0)$. By the comparison theorem for impulsive differential equations, the solutions of the original system satisfy $0 \le (E, Q, P, A) \le (E_s, Q_s, P_s, A_s)$, leading to the conclusion that the pest population vanishes as $t \to \infty$. Consequently, $L(t)$ converges to $\widetilde{L}(t)$ as shown in Lemma 5, confirming that $(0, 0, 0, 0, \widetilde{L}(t))$ is globally asymptotically stable.
\end{pf}

\subsection[\appendixname~\thesubsection]{Theorem 13}\label{app14}

\begin{theorem}\label{th13}
    If $\mathcal{R}_{imp} > 1$, system \eqref{equation5} is permanent, meaning there exists a positive constant $\epsilon^* > 0$ such that every solution $(E(t), Q(t), P(t), A(t), L(t))$ with positive initial conditions satisfies $\liminf_{t \to \infty} E(t) \ge \epsilon^*, \dots, \liminf_{t \to \infty} A(t) \ge \epsilon^*$.
\end{theorem}

\begin{pf}
The permanence of system \eqref{equation5} is established by demonstrating that the pest-extinction periodic solution $(\dots, 0, \widetilde{L}(t))$ is an isolated repeller. When $\mathcal{R}_{imp} > 1$, the spectral radius of the monodromy matrix $M(T)$ for the linearized pest sub-system exceeds unity, implying that the pest-free state is unstable.

Suppose the system is not permanent. Then, the solution trajectory could stay arbitrarily close to the boundary periodic solution for an extended period. However, we can choose a sufficiently small $\epsilon_0 > 0$ such that the perturbed threshold remains greater than one:
$$
    \mathcal{R}_{imp}(\epsilon_0) = \frac{b \alpha_1 \alpha_2 \alpha_3}{\mu_4 (\alpha_3 + \mu_3) (\alpha_2 + \mu_2) (\alpha_1 + \mu_1 + \beta (\langle L \rangle + \epsilon_0))} > 1
$$
Consider the parasitoid population $L(t)$ when the egg density $E(t)$ is small ($E < \epsilon$). The dynamics of $L(t)$ are governed by $\dot{L} = (c\beta E - \mu_L - \omega)L$. For sufficiently small $E$, $L(t)$ remains bounded above by $\widetilde{L}(t) + \epsilon_0$ after a certain time $T_0$. 

Substituting this upper bound into the pest sub-system, we obtain a linear comparison system:
\begin{align*}
\dot{E}_L &= bA_L - (\alpha_1 + \mu_1 + \beta(\widetilde{L}(t) + \epsilon_0))E_L, \\
\dot{Q}_L &= \alpha_1 E_L - (\alpha_2 + \mu_2)Q_L - \eta Q_L^2.
\end{align*}
Since $\mathcal{R}_{imp}(\epsilon_0) > 1$, the zero solution of this comparison system is unstable, and its trajectories will grow exponentially away from the origin. By the impulsive comparison theorem, the original pest densities $(E, Q, P, A)$ must also increase and eventually exit the $\epsilon$-neighborhood of the pest-free periodic solution. 

Using the theory of uniform permanence for impulsive semi-dynamical systems, we conclude that the pest-extinction periodic solution is an isolated invariant set in the boundary of the positive orthant $\mathbb{R}^5_+$. Since there are no other invariant sets on the boundary (as established in Lemma 3 and Lemma 5), and the boundary is repelling, system \eqref{equation5} is permanent. This implies that if the release dosage $\delta$ is insufficient, the soybean pod borer will persist at a strictly positive density.
\end{pf}

\bibliographystyle{cas-model2-names}
\bibliography{cas-refs}
\end{document}